%% file: nonlinearDL.tex
\documentclass[journal,12pt,onecolumn,draftclsnofoot]{IEEEtran}

\usepackage{cite}

\ifCLASSINFOpdf
\else
\fi
\usepackage{amsmath,graphicx}

\usepackage{amsfonts}
\usepackage{hyperref} 
\usepackage{dsfont} 
\usepackage{amssymb} 
\usepackage{bm} 
\usepackage{bbm} 
\usepackage{algorithm}
\usepackage[noend]{algpseudocode}
\usepackage{caption}
\usepackage{subcaption}
\usepackage{mathtools}
\usepackage{amsthm} 
\usepackage{color}
\usepackage{tikz}

\def\x{{\mathbf x}}
\def\z{{\mathbf z}}
\def\y{\mathbf{y}}

\def\d{\mathbf{d}}
\def\u{\mathbf{u}}

\def\h{\mathbf{h}}
\def\l{\mathbf{l}}
\def\n{\mathbf{n}}

\def\D{{\mathbf D}}
\def\X{{\mathbf X}}
\def\A{{\mathbf A}}
\def\M{{\mathbf M}}
\def\I{{\mathbf I}}

\def\R{{\mathbb R}}

\def\0{{\mathbf 0}}

\DeclareMathOperator{\Mr}{\mathbf{M^r}}

\DeclareMathOperator{\Mcp}{\mathbf{M^{c+}}}
\DeclareMathOperator{\Mcm}{\mathbf{M^{c-}}}
\DeclareMathOperator{\bfalpha}{\bm{\alpha}}

\DeclareMathOperator*{\argmin}{argmin}
\DeclareMathOperator*{\sign}{sign}
\DeclareMathOperator*{\dEU}{d}

\DeclareMathOperator*{\cl}{cl}
\DeclareMathOperator{\prox}{prox}

\newcommand{\Lcal}{\mathcal{L}_f}
\newcommand{\Cyt}{f^{-1}\{\y_t\}}
\newcommand{\Cy}{f^{-1}\{\y\}}
\newcommand{\Proj}{\Pi_{\Cy}}

\DeclarePairedDelimiter\floor{\lfloor}{\rfloor}

\newtheorem{theorem}{Theorem}
\newtheorem{prop}{Proposition}

\begin{document}
%
\title{Sparse Recovery and Dictionary Learning from Nonlinear Compressive Measurements}
%
%
%


\author{Lucas Rencker\thanks{L. Rencker, W. Wang and M. D. Plumbley are with the Centre for Vision, Speech and Signal Processing (CVSSP), University of Surrey, Guildford, UK.}, Francis Bach\thanks{F. Bach is with the INRIA-Sierra project-team, Paris, France.}, Wenwu Wang, Mark D. Plumbley\thanks{The research leading to these results has received funding from the European Union's H2020 Framework Programme (H2020-MSCA-ITN-2014) under grant agreement no 642685 MacSeNet.}}

\maketitle

\begin{abstract}
Sparse coding and dictionary learning are popular techniques for linear inverse problems such as denoising or inpainting. However in many cases, the measurement process is nonlinear, for example for clipped, quantized or 1-bit measurements. These problems have often been addressed by solving constrained sparse coding problems, which can be difficult to solve, and assuming that the sparsifying dictionary is known and fixed. Here we propose a simple and unified framework to deal with nonlinear measurements. We propose a cost function that minimizes the distance to a convex feasibility set, which models our knowledge about the nonlinear measurement. This provides an unconstrained, convex, and differentiable cost function that is simple to optimize, and generalizes the linear least squares cost commonly used in sparse coding. We then propose proximal based sparse coding and dictionary learning algorithms, that are able to learn directly from nonlinearly corrupted signals. We show how the proposed framework and algorithms can be applied to clipped, quantized and 1-bit data. 
\end{abstract}


%
\IEEEpeerreviewmaketitle

\section{Introduction}
%
%
%
%

\IEEEPARstart{S}{parse} decomposition and dictionary learning are popular techniques for linear inverse problems in signal processing, such as denoising \cite{2006_Elad_Image,2008_Mairal_Sparse}, inpainting \cite{2005_Elad_Simultaneous,2012_Adler_Audio} or super-resolution \cite{2010_Yang_Image,2012_Yang_Coupled}. Sparse coding aims at finding a sparse set of coefficients $\bfalpha \in \mathbb{R}^{M}$ that accurately represents a signal $\x \in \mathbb{R}^{N}$ from a fixed overcomplete dictionary $\D \in \mathbb{R}^{N \times M}$, and is often formulated as:
\begin{equation}\label{eq:SC}
\min_{\bm{\alpha}} \|\x-\D\bfalpha\|_2^2 \quad + \lambda \Psi(\bfalpha),
\end{equation}
where $\Psi(\cdot)$ is a sparsity inducing regularizer, such as the $\ell_0$ pseudo-norm or the $\ell_1$-norm. Dictionary learning on the other hand, jointly learns the dictionary $\D$ and sparse coefficients $\bfalpha_t$ from a set of training signals $\{\x_t\}_{t=1,\dotsc,T}$:
\begin{equation}\label{eq:DL}
\min_{\D, \bm{\alpha}_t} \sum_{t=1}^T \left[ \|\x_t-\D\bfalpha_t\|_2^2 \quad + \lambda \Psi(\bfalpha_t) \right].
\end{equation}
However, the observed signals are often distorted or measured in a nonlinear way:
\begin{equation}
\y = f(\x),
\end{equation}
where $f$ is a nonlinear measurement function, and $\x$ is the original (unknown) clean signal. Examples of nonlinear distortions include \emph{clipping} (or saturation) and \emph{quantization}. Clipping is often due to dynamic range limitations in acquisition systems, when a signal reaches a maximum allowed amplitude, and the waveform is truncated above that threshold \cite{2010_Mansour_Color,2011_Laska_Democracy,2011_Adler_constrained,2013_Kitic_Consistent,2013_Defraene_Declipping,2014_Siedenburg_Audio,2015_Kitic_Sparsity,2016_Foucart_Sparse}. Quantization is a common process in analog-to-digital conversion that maps a signal from a continuous input space to a (finite) discrete space \cite{1998_Gray_Quantization}. More recently, 1-bit compression has attracted a lot of interest, as an extreme quantization scheme where samples are coded using only one bit per sample \cite{2008_Boufounos_1}, i.e. only measuring the signs of the signal. Clipping and quantization are non-linear, non-smooth, and \emph{compressive} measurements, i.e. the measurement map is non-invertible. For these reasons, the recovery of clipped/quantized signals is a challenging problem.

Recovering a signal from clipped or quantized measurements can be treated as linear inverse problems, by simply ignoring the nonlinearities, i.e. treating clipped samples as missing \cite{2009_Laska_Finite,2011_Laska_Democracy} and quantization error as additive noise \cite{2006_Candes_Near}. Similarly, 1-bit signals can be tackled by using the sign measurements directly as an input \cite{2008_Boufounos_1,2009_Boufounos_Greedy}. However using a formulation that is \emph{consistent} with the measurement process, i.e. that takes into account our knowledge about the nonlinear measurement function, has been shown to greatly improve the reconstruction \cite{2010_Mansour_Color,2011_Laska_Democracy,2011_Adler_constrained,2012_Adler_Audio,2013_Kitic_Consistent,2013_Defraene_Declipping,2014_Siedenburg_Audio,2015_Kitic_Sparsity,2016_Foucart_Sparse,2009_Dai_Distortion,2010_Zymnis_Compressed,2011_Jacques_Dequantizing,2013_Jacques_Quantized,2016_Moshtaghpour_Consistent,2008_Boufounos_1,2009_Boufounos_Greedy,2013_Jacques_Robust}. Specially tailored cost functions, constraints, or regularizers have independently been proposed to deal with clipped \cite{2010_Mansour_Color,2011_Laska_Democracy,2011_Adler_constrained,2012_Adler_Audio,2013_Kitic_Consistent,2013_Defraene_Declipping,2014_Siedenburg_Audio,2015_Kitic_Sparsity,2016_Foucart_Sparse}, quantized \cite{ 2009_Dai_Distortion,2010_Zymnis_Compressed,2011_Jacques_Dequantizing,2013_Jacques_Quantized,2016_Moshtaghpour_Consistent} or 1-bit \cite{2008_Boufounos_1,2009_Boufounos_Greedy,2013_Jacques_Robust} measurements. These formulations often involve solving constrained sparse coding problems, which can be difficult and computationally expensive to solve, since they involve 
computing expensive non-orthogonal projections at each iteration \cite{2015_Kitic_Sparsity,2016_Moshtaghpour_Consistent}.

Reconstruction methods proposed in the literature assume that the signal is sparse in some orthogonal basis \cite{2010_Zymnis_Compressed, 2009_Dai_Distortion,2011_Jacques_Dequantizing,2013_Jacques_Quantized,2016_Moshtaghpour_Consistent,2007_Boufounos_Sparse,2008_Boufounos_1,2009_Boufounos_Greedy,2013_Jacques_Robust,2013_Jacques_Quantized}, or in a fixed dictionary \cite{2011_Adler_constrained,2012_Adler_Audio,2013_Kitic_Consistent,2015_Kitic_Sparsity}. However it has been shown in a range of applications (when the measurements are linear), that \emph{learning} the dictionary from the observed data greatly improves the reconstruction compared to using a fixed dictionary \cite{2006_Elad_Image,2008_Mairal_Sparse,2011_Couzinie-Devy_Dictionary,2014_Mairal_Sparse}. 

\subsection{Contributions, and main results}

Our contributions are as follows: 
\begin{itemize}
\item We propose a unifying framework for signal recovery from nonlinear measurements such as clipping, quantization and 1-bit measurements, i.e. addressing these three problems in a unified fashion rather than individually. More specifically, we show how these problems can be formulated as minimizing the distance to a convex feasibility set, which models our assumption about the nonlinear measurement process. In particular, the proposed cost generalizes the linear least squares commonly used in sparse coding, as well as several cost functions proposed independently for declipping and 1-bit recovery.
\item Using properties of projection operators over convex sets, we show that the proposed cost function is continuous, convex and differentiable with Lipschitz gradient. Our main result uses Danskin's Min-Max theorem \cite{1967_Danskin_theory}, that allows us to derive a closed-form gradient for the proposed cost. 
\item We propose proximal-based consistent sparse coding, and dictionary learning algorithms, for nonlinear measurements. We show that these algorithms can be applied to clipped, quantized and 1-bit measurements.
\end{itemize}

\subsection{Organization of paper}

The paper is organised as follows: in Section \ref{sec:background}, we briefly review sparse recovery and dictionary learning from linear measurements, and some strategies proposed to deal with clipped, quantized and 1-bit measurements. In Section \ref{sec:cost} we propose a unifying cost function for nonlinear measurements, and show some of its properties. In Section \ref{sec:algorithms} we propose consistent sparse coding and dictionary learning algorithms using the proposed cost. Applications of the proposed framework, and links to previous work are presented in Section \ref{sec:Applications}. The performance of the proposed algorithm is presented in Section \ref{sec:evaluation}, before the conclusion is drawn. 

\subsection{Notation}

Bold lowercase letters $\x$ denote vectors and bold uppercase letters $\X$ denote matrices. The $i$-th element of a vector $\x$ is noted $x_i$. The identity matrix is noted $\I$. The p-norm of a vector $\x$ is $\|\x\|_p = (\sum x_i^p)^{1/p}$. The $\ell_0$ pseudo-norm (i.e. the number of non-zero elements) of $\x$ is noted $\|\x\|_0$. For a matrix $\X$, $\|\X\|_2$ denotes the matrix 2-norm, i.e. the largest singular value of $\X$. We denote $(\x)_+ = \max(\mathbf{0},\x)$ (where $\max$ is the element-wise maximum), and $(\x)_- = -(-\x)_+$. The floor (i.e. closest lower integer) of a vector is noted $\floor{\x}$ (applied element-wise). The sign (positive or negative) of each element of $\x$ is noted $\sign(\x)$. The element-wise multiplication is $\odot$. For a set $\mathcal{C}$, $\cl(\mathcal{C})$ is the closure of $\mathcal{C}$, and $\mathbbm{1}_{\mathcal{C}}(\cdot)$ is the indicator function of that set, i.e. $\mathbbm{1}_{\mathcal{C}}(\x) = 0$ when $\x \in \mathcal{C}$, $+\infty$ otherwise. The notation $\preceq$ denotes vector-wise inequality.

\section{Background}
\label{sec:background}

In this paper, we denote observation vectors as $\y \in \mathbb{R}^L$, with $\y = f(\x)$, where $\x \in \mathbb{R}^N$ is the original un-observed clean signal, and $f$ is a measurement or distortion function. We further assume that the signal $\x$ can be decomposed as $\x = \D \bfalpha$ with $\D \in \mathbb{R}^{N \times M}$ an \emph{overcomplete} dictionary of $M$ atoms ($N<M$), and $\bfalpha \in \mathbb{R}^M$ is a sparse activation vector. In this section we review the different types of linear and nonlinear measurement functions $f$, and the associated problem formulations appearing in the literature. 

\subsection{Sparse coding from linear measurements}

A widely studied case is when the measurement function is linear, i.e. $f(\x) = \M\x$ with $\M \in \mathbb{R}^{L \times N}$. The corresponding sparse coding problem is often formulated as:
\begin{equation}\label{eq:linear_SC}
\min_{\bfalpha} \|\y-\M\D\bfalpha\|_2^2+\lambda\Psi(\bfalpha).
\end{equation}
For example, $\M = \I$ (the identity matrix) corresponds to clean signals or signals subject to additive Gaussian noise \cite{2006_Elad_Image,2008_Mairal_Sparse}. When $\M$ is a diagonal binary matrix, \eqref{eq:linear_SC} corresponds to an inpainting problem \cite{2005_Elad_Simultaneous,2012_Adler_Audio}. 

\subsection{Sparse coding from nonlinear measurements}

Sparse coding from linear measurements have been extensively studied in the literature. Signal acquisition systems however often measure signals in a \emph{nonlinear} way, for examples in saturated, quantized or 1-bit measurements. In the following, we review reconstruction strategies proposed in the literature for saturated, quantized and 1-bit measurements.



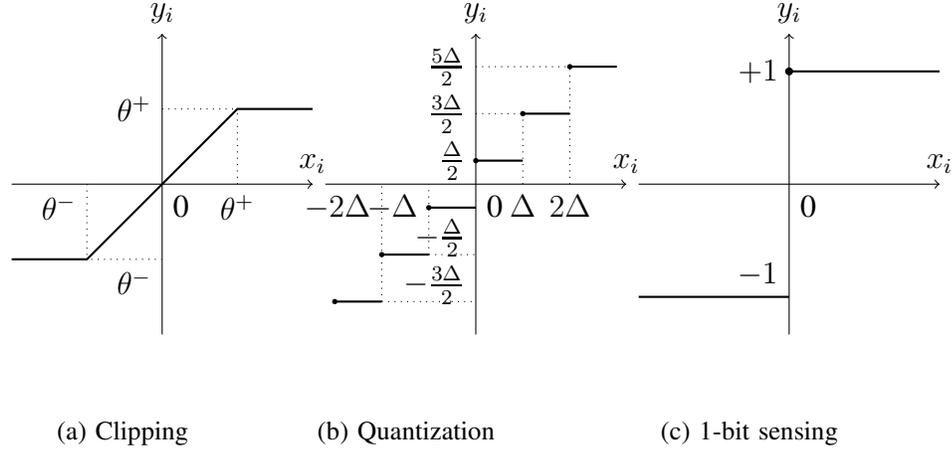
\begin{figure}[h]
	\centering
	\begin{subfigure}[h!]{0.18\textwidth}
		\input{clipping_func.tikz}
		\caption{Clipping}
		\label{fig:clipping_function}
	\end{subfigure}
	\hspace{0.5cm}
	\begin{subfigure}[h!]{0.18\textwidth}
		\input{quantization_func.tikz}
		\caption{Quantization}
		\label{fig:quantization_function}
	\end{subfigure}
	\hspace{1.3cm}
	\begin{subfigure}[h!]{0.18\textwidth}
		\input{1bit_func.tikz}
		\caption{1-bit sensing}
		\label{fig:1bit_function}
	\end{subfigure}
	\caption{Visualization of different nonlinear measurement functions $f$ (output $y_i = f(x_i)$ versus input $x_i$).}
	\label{fig:nonlinear_functions}
\end{figure}


\subsubsection{Clipped measurements}

We consider the case of hard clipping, where each sample $x_i$ is measured as:
\begin{equation}
y_i = f(x_i) = 
\begin{cases}
\theta^+ \quad \text{if} \quad x_i \geq \theta^+\\
\theta^- \quad \text{if} \quad x_i \leq \theta^-\\
x_i \quad \text{otherwise,}
\end{cases}
\end{equation}
where $\theta^+ > \theta^-$ are positive and negative clipping thresholds respectively (Figure \ref{fig:clipping_function}). This can be written in vector form as:
\begin{equation}\label{eq:clipping_function}
\y = f(\x) = \Mr \x + \theta^+ \Mcp \mathbf{1} + \theta^- \Mcm \mathbf{1},
\end{equation}
where $\mathbf{1}$ is the all-ones vector in $\mathbb{R}^{N}$, and $\Mr, \Mcp$ and $\Mcm$ are diagonal binary sensing matrices, that define the \emph{reliable} (i.e. unclipped), positive and negative clipped samples respectively and such that $\Mr + \Mcp +\Mcm = \I$. In practice, the clipping thresholds can be estimated from the measurement $\y$ as (e.g.) $\theta^+ = \max_i(y_i)$, and the sensing matrices as $[\Mcp]_{i,i} = 1$ if $y_i = \theta^+$, $0$ otherwise.

Declipping can be treated as a linear inverse problem by discarding the clipped samples, treating declipping as a linear \emph{inpainting} problem, i.e. solving \eqref{eq:linear_SC} with $\M = \Mr$ \cite{2009_Laska_Finite,2011_Laska_Democracy,2012_Adler_Audio}. However, the reconstruction can be improved by adding extra knowledge about the clipping process. Indeed, we know that the clipped samples should have an amplitude that is greater than the clipping threshold. This extra information can be enforced by solving the following constrained problem \cite{2009_Smaragdis_Dynamic,2010_Mansour_Color,2011_Adler_constrained,2013_Defraene_Declipping,2016_Foucart_Sparse}:
\begin{equation}\label{eq:constrained_declipping2}
\min_{\bfalpha} \Psi(\bfalpha) + \mathbbm{1}_{\mathcal{C}(\y)}(\D\bfalpha),
\end{equation}
where:
\begin{align}
\begin{split}
\mathcal{C}(\y) \triangleq \{\x| \Mr& \y=\Mr \x, \Mcp \x \succeq \theta^+\Mcp \mathbf{1}, \\
& \Mcm \x \preceq \theta^-\Mcm \mathbf{1}\}
\end{split}
\end{align}
is the \emph{clipping consistency} set. Eqn. \eqref{eq:constrained_declipping2} is a constrained, non-smooth and possibly non-convex sparse decomposition problem, which can be difficult to solve. An alternating direction method of multipliers (ADMM) \cite{2011_Boyd_Distributed} based algorithm was proposed in \cite{2015_Kitic_Sparsity} to solve \eqref{eq:constrained_declipping2}. When the dictionary is a tight frame, the algorithm can be computed efficiently \cite{2018_Zaviska_Revisiting}. However, for general dictionaries, the algorithm is computationally expensive, since it involves computing non-orthogonal projections at each iteration \cite{2015_Kitic_Sparsity}. A soft consistency metric was used in \cite{2011_Laska_Democracy,2013_Kitic_Consistent,2014_Siedenburg_Audio}:
\begin{equation}\label{eq:consistentIHT}
\begin{split}
\min_{\bfalpha} \frac{1}{2}\big[\|\Mr&(\y-\D\bfalpha)\|_2^2 + \|\Mcp (\theta^+\mathbf{1}-\D\bfalpha)_+\|_2^2\\ + & \|\Mcm (\theta^-\mathbf{1}-\D\bfalpha)_-\|_2^2\big] +\lambda\Psi(\bfalpha),
\end{split}
\end{equation}
The data-fidelity term in \eqref{eq:consistentIHT} is convex and smooth, so methods based on iterative hard thresholding \cite{2009_Blumensath_Iterative,2013_Kitic_Consistent} or proximal algorithms \cite{2014_Parikh_Proximal,2014_Siedenburg_Audio} can directly be applied. 

\subsubsection{Quantized measurements}

Quantization maps a continuous input space onto a finite discrete set of codewords $\mathcal{Y} = \{y_1, ..., y_p\}$. A quantization map $f$ is defined by a set of quantization levels $\mathcal{R}_q = [l_q, u_q)$ and the relation $x \in \mathcal{R}_q \Leftrightarrow f(x) = y_q$, i.e. samples that fall into $\mathcal{R}_q$ are quantized as $y_q$. For example in the case of a uniform mid-riser quantizer, $\mathcal{R}_q = [\Delta q, \Delta(q+1))$, and the quantization function can be written as:
\begin{equation}
f(\x) = \Delta \left\lfloor\frac{\x}{\Delta}\right\rfloor + \frac{\Delta}{2},
\end{equation}
where $\Delta > 0$ is the quantization bin width (Figure \ref{fig:quantization_function}).

De-quantization can be treated as a simple linear inverse problem by considering quantization error as additive noise, and using a linear sparse model \eqref{eq:linear_SC} \cite{2006_Candes_Near}. However it has been shown that using a more accurate model of the quantization process improves the reconstruction. Bayesian approaches \cite{2010_Zymnis_Compressed}, $\ell_p$-based data-fidelity terms \cite{2011_Jacques_Dequantizing}, or specially-tailored cost functions \cite{2013_Jacques_Quantized} have been proposed in the literature to enforce quantization consistency. Constrained formulations were proposed in \cite{2009_Dai_Distortion,2016_Moshtaghpour_Consistent} in order to enforce consistency:
\begin{equation}\label{eq:constrained_quantization}
\min_{\bfalpha} \Psi(\bfalpha) + \mathbbm{1}_{\mathcal{R}}(\D\bfalpha)
\end{equation}
where $\mathcal{R} = \mathcal{R}_{q_1} \times ... \times \mathcal{R}_{q_N}$ and $\mathcal{R}_{q_i}$ is the quantization region associated with the $i$-th sample $y_i$. However similarly to the constrained declipping scenario \eqref{eq:constrained_declipping2}, solving \eqref{eq:constrained_quantization} can be computationally expensive.

\subsubsection{1-bit measurements}

1-bit measurement can be seen as an extreme quantization using only one bit per sample, or similarly an extreme saturation where the clipping level tends to zero:
\begin{equation}
f(\x) = \sign(\x).
\end{equation}
1-bit measurement consistency has also been proposed in \cite{2008_Boufounos_1,2009_Boufounos_Greedy}:
\begin{equation}\label{eq:1bit_l2}
\min_{\bfalpha} \frac{1}{2}\|\big(\y\odot(\D\bfalpha)\big)_-\|_2^2 + \lambda \Psi(\bfalpha)
\end{equation}

\subsection{Dictionary learning}

Dictionary learning from linear measurements can be formulated as \cite{2014_Mairal_Sparse}:
\begin{equation}\label{eq:linear_DL}
\min_{\D \in \mathcal{D}, \bm{\alpha}_t}  \sum_{t=1}^T \big[\|\x_t-\D\bfalpha_t\|_2^2+\lambda \Psi(\bfalpha_t)\big]
\end{equation}
where $\{\x_t\}_{1...T}$ is a collection of $T$ signals in $\mathbb{R}^N$, and $\bfalpha_t$ are the corresponding sparse activation vectors. The dictionary $\D$ is often constrained to be in the convex set $\mathcal{D} = \{\D \in \mathbb{R}^{N \times M}|\forall i, \|\d_i\|_2 \leq 1\}$ in order to avoid scaling ambiguity \cite{2014_Mairal_Sparse}. Many algorithms have been proposed in the literature to solve \eqref{eq:linear_DL}, such as MOD \cite{1999_Engan_Method}, K-SVD \cite{2006_Elad_Image} stochastic gradient descent \cite{1997_Olshausen_Sparse,2009_Mairal_Online}, or SimCO \cite{2012_Dai_Simultaneous}. 

Dictionary learning for 1-bit data have recently been addressed in \cite{2016_Zayyani_Dictionary,2016_Shen_Online}. To our knowledge, dictionary learning from saturated and quantized measurements, however, has not been addressed in the literature\footnote{We presented some preliminary results on dictionary learning for declipping in \cite{2018_Rencker_Consistent}}. In the next sections, we propose a unifying and computationally tractable framework for sparse coding and dictionary learning from nonlinear measurements.

\section{A unifying framework for nonlinear signal reconstruction}
\label{sec:cost}

Let $f:\mathcal{X} \mapsto f(\mathcal{X}) = \mathcal{Y}$ be an arbitrary - and possibly nonlinear - measurement function from a clean input space $\mathcal{X}$ to a measurement space $\mathcal{Y}$. For a measured signal $\y \in \mathcal{Y}$, we propose a cost function (or data-fidelity term) defined for all $\x \in \mathcal{X}$ as:
\begin{equation}\label{eq:cost}
\Lcal(\x, \y) = \dEU(\x, \Cy)
\end{equation}
where $\Cy$ is the \emph{pre-image} of $\{\y\}$ under the measurement map $f$:
\begin{equation}
\Cy \triangleq \{\x \in \mathcal{X}|f(\x) = \y\},
\end{equation}
and $\dEU(\x,\mathcal{C})$ is the distance between $\x$ and the set $\mathcal{C}$, defined for a (pointwise) distance metric $\dEU(\cdot, \cdot)$ as:
\begin{equation}\label{eq:distance}
\dEU (\x, \mathcal{C}) \triangleq \inf_{\z \in \mathcal{C}} \dEU(\x,\z).
\end{equation}
The set $\Cy$ can be seen as a \emph{feasibility set}, i.e. the set of all possible input signals $\x \in \mathcal{X}$ that could have generated $\y$ when measured through $f$. The cost \eqref{eq:cost} thus measures how ``close" a signal $\x$ is to the feasibility set associated with the measurement $\y$. Minimizing \eqref{eq:cost} thus promotes consistency since it minimizes the distance between an estimate $\x$ and its feasibility set $\Cy$. However unlike constrained formulations \eqref{eq:constrained_declipping2} or \eqref{eq:constrained_quantization}, here measurement-consistency is enforced in a simple unconstrained way. 

\subsection{Assumptions on $f$, and choice of distance}

Without any assumptions on the feasibility sets $\Cy$ and the metric $\d(\cdot,\cdot)$, $\x \mapsto \Lcal(\x,\y)$ is in general non-convex and non-smooth, and therefore difficult to optimize. However we show here that under certain conditions, the proposed cost \eqref{eq:cost} exhibits convenient properties such as convexity, and differentiability with Lipschitz gradient.

The first assumption is that for all $\y \in \mathcal{Y}$, the pre-image set $\Cy$ is convex. This assumption is verified by many measurement functions $f$ found in practice, such as linear measurements, and nonlinear measurements such as clipping, quantization and 1-bit (see Section \ref{sec:Applications}). For separable functions $f(\x) = \left[f_1(x_1), \dotsc, f_N(x_N)\right]$, a simple sufficient condition such that $\Cy$ is convex for all $\y$ is that each $f_i(\cdot)$ is monotonic (but not necessarily strictly monotonic). This condition is however not necessary, and the set of functions with convex pre-images includes a wider range of measurement functions. Convexity of the pre-image sets ensures that the proposed cost \eqref{eq:cost} is convex and differentiable, as will be shown in this section.

Various metrics can be chosen for the pointwise distance in \eqref{eq:distance}. A popular metric to measure point-wise distances is the squared Euclidean distance $\|\cdot\|_2^2$. The squared Euclidean distance, or least squares, is popular in sparse coding and dictionary learning, in part due to its convexity and Lipschitz differentiability. Other pointwise distances are available, such as $p$-norms $\|\cdot\|_p^p$ with $0 < p \leq 1$, which might be more appropriate depending on the data at hand. However these come at a cost of non-differentiability and/or non-convexity, which often lead to more difficult optimization problems. In this paper, for simplicity, and in order to favour computationally efficient methods, we focus on the least-squares distance. In fact, we show in this section that the proposed cost \eqref{eq:cost} benefits from the same properties as -- and naturally extends -- the linear least-squares cost commonly used in sparse coding.

In the remainder of this paper, we assume $\mathcal{X} = \mathbb{R}^N$, $\d(\x,\y) = \frac{1}{2}\|\x-\y\|_2^2$, and $\Cy$ convex for all $\y \in \mathcal{Y}$. Note that for a set $\mathcal{C}$, $\dEU(\x,\mathcal{C}) = \dEU(\x,\cl(\mathcal{C}))$, so we can assume without loss of generality that $\Cy$ is closed. Note also that since $\mathcal{Y} = f(\mathcal{X})$, $\Cy$ is non-empty for all $\y \in \mathcal{Y}$.

\subsection{Properties of the proposed cost function}

We consider a fixed $\y \in \mathcal{Y}$, and review the properties of $\x \mapsto \Lcal(\x,\y)$. Since $\Cy$ is non-empty, closed and convex, the Projection Theorem (Appendix \ref{app:projection_theorem}) ensures existence and uniqueness of a minimizer $\z^*$ of $\|\x-\z\|^2_2$ in $\Cy$. This minimizer is defined as the orthogonal projection $\Proj(\x)$ of $\x$ on the set $\Cy$. In particular, the infinimum in \eqref{eq:distance} is attained, and $\Lcal(\x,\y)$ can be redefined as:
\begin{equation}
\Lcal(\x,\y) = \frac{1}{2}\|\x-\Proj(\x)\|_2^2.
\end{equation}
In addition, the continuity property of the projection operator on a convex set ensures that $\Lcal(\cdot,\y)$ is a continuous function (as a composition of continuous functions). We now present some properties of $\Lcal(\cdot,\y)$, which make it suitable for a range of optimization algorithms:

\begin{prop}\label{prop:convexity}
$\Lcal(\cdot,\y)$ is a convex cost function.
\end{prop}

\begin{prop}\label{prop:smoothness}
$\Lcal(\cdot,\y)$ is differentiable, with gradient:
\begin{equation}
\nabla_\x \Lcal(\x,\y) = \x-\Pi_{\Cy}(\x).
\end{equation}
\end{prop}
\begin{prop}\label{prop:lipschitz}
The gradient $\nabla_{\x} \Lcal(\x, \y)$ is Lipschitz continuous with constant $L = 1$, i.e. for all $\x_1, \x_2$:
\begin{equation}
\|\nabla_{\x}\Lcal(\x_1, \y)-\nabla_{\x}\Lcal(\x_2, \y)\|^2_2 \leq \|\x_1-\x_2\|_2^2
\end{equation}
\end{prop}
Proposition \ref{prop:convexity} is due to the convexity of the set $\Cy$ and of the least-squares cost. Proposition \ref{prop:smoothness} is a direct consequence of Danskin's Min-Max theorem (Appendix \ref{app:Danskin}) and of the uniqueness of the projection operator. Proposition \ref{prop:lipschitz} is a consequence of the contraction property of projection onto convex sets. See Appendix \ref{app:proofs} for more detailed proofs.

Continuity, convexity and Lipschitz differentiability makes the proposed cost function suitable for a range of optimization algorithm as will be seen in the next section. Moreover, when $f$ is the identity map ($f(\x) = \x$), we have $\Cy = \y$ and $\Lcal(\x,\y) = \frac{1}{2}\|\x-\y\|_2^2$. The proposed cost thus generalizes the least squares cost commonly used in sparse coding and dictionary learning. We will show in Section \ref{sec:Applications} how the proposed cost also generalizes several cost functions proposed in the literature for inpainting, declipping and 1-bit signals, and how it can be applied to quantized measurements. The proposed cost thus provides a unifying framework to tackle all these problems. In the next section we propose simple proximal-based algorithms for sparse coding and dictionary learning using the proposed cost.

\section{Proposed consistent sparse coding and dictionary learning algorithms}
\label{sec:algorithms}

\subsection{Sparse coding algorithms}

For a nonlinear observation $\y$ and a fixed dictionary $\D$, we propose to formulate consistent sparse coding as:
\begin{equation}\label{eq:SC_formulation}
\min_{\bm{\alpha}} \Lcal(\D\bfalpha,\y) + \lambda \Psi(\bfalpha).
\end{equation}
Solving \eqref{eq:SC_formulation} is thus a problem of minimizing the sum of a convex, smooth cost function and a non-smooth regularizer. When the regularizer $\Psi(\cdot)$ is convex, such as the $\ell_1$-norm, this can be classically optimized using \emph{proximal} descent algorithms \cite{2014_Parikh_Proximal, 2011_Combettes_Proximal}. The proposed proximal-based sparse coding algorithm is presented in Algorithm \ref{alg:SC}.
\begin{algorithm}[H]
	\caption{Proposed consistent sparse coding algorithm (fixed parameter $\lambda$)}\label{alg:SC}
	\begin{algorithmic}
		\Require $f, \y, \D, \bfalpha^0, \lambda, \mu_1$
		\State initialize: $\bfalpha \gets \bfalpha^0$
		\While{stopping criterion not reached}
		\State \textbf{Gradient descent step:}
		\State $\bfalpha \gets \bfalpha + \mu_1 \D^T (\Pi_{\Cy}(\D\bfalpha)-\D\bfalpha)$
		\State \textbf{Proximal thresholding:}
		\State $\bfalpha \gets \prox_{\lambda \Psi}(\bfalpha) \triangleq \argmin_{\u}\|\u-\bfalpha\|_2^2 + \lambda \Psi(\bfalpha)$
		\EndWhile
		\State \textbf{return} $\hat{\bfalpha}$
	\end{algorithmic}
\end{algorithm}
Note that Algorithm \ref{alg:SC} is presented here in its simplest form, however it can easily be accelerated using the same strategy as in \cite{2009_Beck_fast}. 
\subsubsection{Convergence}
The convergence properties of Algorithm \ref{alg:SC} can be summarized as follows:
\begin{prop}\label{prop:convergence}
	For a convex penalty $\Psi(\cdot)$ and a step size $0< \mu_1 \leq \frac{1}{\|\D\|_2^2}$ (where $\|\D\|_2$ is the highest singular value of $\D$), Algorithm \ref{alg:SC} converges. 
\end{prop}
\begin{proof}
	This is a direct consequence of Propositions \ref{prop:convexity} - \ref{prop:lipschitz} and classical results on proximals algorithms with convex and Lipschitz differentiable functions (see, \cite{2011_Combettes_Proximal,2014_Parikh_Proximal} or \cite[Theorem 3.1.]{2009_Beck_fast} for a proof). Here $\|\D\|_2^2$ is the Lipschitz constant of $\bfalpha \mapsto \Lcal(\D\bfalpha,\y)$, as a consequence of Proposition \ref{prop:lipschitz}.
\end{proof}

\subsubsection{Measurement consistency}

The regularization parameter $\lambda >0$ in \eqref{eq:SC_formulation} controls a tradeoff between sparsity and consistency. There is however no analytical formula for the relationship between the parameter $\lambda$, and how ``consistent" the resulting signal is (i.e., how close it is from its consistency set). In particular the resulting signal is not guaranteed to be \emph{exactly} consistent, i.e. exactly within its consistency set. One way to circumvent this is to solve \eqref{eq:SC_formulation} iteratively for different values of the parameter $\lambda$, starting from a large value, and progressively lowering $\lambda$ until a desired consistency is achieved. This can be done efficiently using a warm-start strategy, i.e. initializing each iteration by the estimate of the previous iteration. Algorithm \ref{alg:SC_paramfree} is a simple modification of Algorithm \ref{alg:SC} that implements this strategy ($\{\lambda^k\}_{k\geq0}$ is a series of non-increasing and strictly positive values).
\begin{algorithm}[H]
	\caption{Proposed consistent sparse coding algorithm (adaptive parameter $\lambda$)}\label{alg:SC_paramfree}
	\begin{algorithmic}
		\Require $f, \y, \D, \bfalpha^0, \{\lambda^k\}_{k\geq0}, \mu_1, \epsilon$
		\State initialize: $\bfalpha \gets \bfalpha^0, k \gets 0, \lambda \gets \lambda^0$
		\While{$\Lcal(\D\bfalpha^k,\y) > \epsilon$}
		\State \textbf{Iterate until convergence:}
		\State $\bfalpha \gets \prox_{\lambda \Psi}(\bfalpha + \mu_1 \D^T (\Pi_{\Cy}(\D\bfalpha)-\D\bfalpha))$
		\State \textbf{Update $\lambda$:}
		\State $\lambda \gets \lambda^{k+1}, \ k \gets k+1$

		\EndWhile
		\State \textbf{return} $\hat{\bfalpha}$
	\end{algorithmic}
\end{algorithm}
Algorithm \ref{alg:SC_paramfree} is similar to \emph{homotopy} methods proposed for sparse coding \cite{2005_Malioutov_Homotopy,2004_Efron_Least}. By progressively decreasing $\lambda$, the algorithm essentially adds more coefficients to the support set, making the signal less sparse but more consistent. We furthermore have the following: 
\begin{prop}\label{prop:paramfree}
	Let $F(\lambda, \bfalpha) \triangleq \Lcal(\D\bfalpha,\y) + \lambda \Psi(\bfalpha)$, $\lambda^0 > \dotsb >\lambda^k > \dotsb > 0$ and $\bfalpha^k \triangleq \argmin_{\bfalpha} F(\lambda^k, \bfalpha)$. We have:

	$\forall k$, 
	$\Psi(\alpha^{k+1}) \geq \Psi(\alpha^{k})$ and $\Lcal(\D\bfalpha^{k+1},\y) \leq \Lcal(\D\bfalpha^k,\y)$.
	
	Moreover, if $\{\lambda_k\} \rightarrow 0$, then the sequence $\{F(\lambda^{k},\bfalpha^{k})\}_{k\geq 0}$ converges to an optimum of the constrained problem:
	\begin{equation}\label{eq:constrained}
	\min_{\bfalpha} \Psi(\bfalpha) \quad \text{s.t.} \quad \D\bfalpha \in \Cy.
	\end{equation}
\end{prop}
Proposition \ref{prop:paramfree} and its proof (in Appendix \ref{app:proofs}) are inspired by results on penalty methods for constrained optimization problems \cite{1999_Bertsekas_Nonlinear}. Algorithm \ref{alg:SC_paramfree} can in fact be seen as an unconstrained penalty method to solve the constrained problem \eqref{eq:constrained}. Proposition \ref{prop:paramfree} shows that Algorithm \ref{alg:SC_paramfree} leads to asymptotically consistent solutions. 


\subsubsection{Influence of noise}

In the case when additive noise is present, the signal is measured as:
\begin{equation}
\y = f(\D\bfalpha^* + \n),
\end{equation}
i.e. $\D\bfalpha^* + \n \in f^{-1}(\y)$, and the original signal $\D\bfalpha^*$ is no longer necessarily in the pre-image set of the observation $\y$. For this reason, finding a solution that is both sparse and consistent with the measurement $\y$ might be infeasible. We can however show that if the noise level is small, the cost function does not deviate far from zero, since:
\begin{equation}
\begin{split}
\Lcal(\D\bfalpha^*, \y) &= \min_{\z \in \Cy} \|\D\bfalpha^* - \z\|_2^2\\
& \leq \min_{\z \in \Cy} \|\D\bfalpha^* +\n - \z\|_2^2 + \|\n\|_2^2\\ 
& = \|\n\|_2^2.
\end{split}
\end{equation}
In other words, if the noise level $\|\n\|_2^2$ is small, then the original signal $\D\bfalpha^*$ is \emph{approximately} consistent with the measurement $\y$, i.e. $\Lcal(\D\bfalpha^*, \y)$ is small. A good reconstruction strategy is then to solve:
\begin{equation}
\argmin_{\bfalpha} \Psi(\bfalpha) \quad \text{s.t.} \quad \Lcal(\D\bfalpha, \y) \leq \|\n\|_2^2,
\end{equation}
which in its Lagrangian form is equivalent to \ref{eq:SC_formulation}.

\subsection{Dictionary learning algorithm}

For a collection $\{\y_t\}_{1,\dotsc,T}$ of $T$ signals measured through the same measurement function $f$, consistent dictionary learning can be formulated using the proposed cost as:
\begin{equation}\label{eq:DL_formulation}
\min_{\D \in \mathcal{D}, \bm{\alpha}_t} \sum_{t=1}^{T} \Big(\Lcal(\D\bfalpha_t,\y_t) + \lambda \Psi(\bfalpha_t)\Big)\\
\end{equation}
Jointly minimizing $\D$ and $\{\bfalpha_t\}_{t=1,...,T}$ in \eqref{eq:DL_formulation} is a non-convex problem. Dictionary learning algorithms typically alternate between a sparse coding, and a dictionary update step \cite{2014_Mairal_Sparse}. The sparse coding step (with a fixed dictionary) can be solved using the proposed consistent sparse coding algorithm (Algorithm \ref{alg:SC}). Once the sparse codes $\{\bfalpha_t\}_{1,\dotsc,T}$ have been updated, the dictionary update step can be formulated as:
\begin{equation}\label{eq:dict_cost}
\min_{\D \in \mathcal{D}} \sum_{t=1}^{T}\Lcal(\D\bfalpha_t,\y_t),
\end{equation}
which can be solved using projected gradient descent \cite{1999_Bertsekas_Nonlinear}, i.e. alternating between a gradient descent step, and a projection step $\Pi_{\mathcal{D}}$ which here simply re-normalizes each column $\d_i$ of $\D$ as $\d_i \gets \d_i/\max(\|\d_i\|_2,1)$.
The proposed dictionary update step is thus similar to classical projected gradient descent approaches already proposed for dictionary learning \cite{1997_Olshausen_Sparse,2008_Aharon_Sparse,2008_Kavukcuoglu_Fast,2010_Mairal_Online}. The proposed dictionary learning algorithm is presented in Algorithm \ref{alg:DL}. The parameter $\mu_2$ is a step size which can be set as $\mu_2 = 1/L_2$ where $L_2 = \|\A\|_2^2$ ($\A \triangleq \left[\bfalpha_1, ..., \bfalpha_T\right]$) is the Lipschitz constant of the cost in \eqref{eq:dict_cost}.

\begin{algorithm}[H]
	\caption{Proposed consistent dictionary learning algorithm}\label{alg:DL}
	\begin{algorithmic}
		\Require $f, \{\y_t\}_{1...T}$, $\D^0$, $\{\bfalpha_t^{0}\}_{1,\dotsc,T}, \lambda$
		\State initialize: $\D^{(0)} \gets \D^0$, $\bfalpha_t^{(0)} \gets \bfalpha_t^{0}$, $i \gets 0$
		\While{stopping criterion not reached}
		\State $i \gets i+1$
		\State \textbf{Sparse coding step:}
		\For{$t = 1 ... T$}
		\State Initialize $\bfalpha_t \gets \bfalpha_t^{(i-1)}$.
			\State Update $\bfalpha_t^{(i)}$ using Algorithm \ref{alg:SC} with $\D = \D^{(i-1)}$.
		\EndFor
		\State \textbf{Dictionary update step:} 
		\State Initialize $\D \gets \D^{(i-1)}$
			\While{not converged}
			\State $\D \gets \D + \mu_2 \sum_{t} (\Pi_{\Cyt}(\D\bfalpha_t^{(i)})- \D\bfalpha_t^{(i)})\bfalpha_t^{(i)T}$
			\State $\D \gets \Pi_{\mathcal{D}}(\D)$
			\EndWhile
		\State $\D^{(i)} \gets \D$
		\EndWhile
		\State \textbf{return} $\hat{\D}, \{\hat{\bfalpha}_t\}_{1...T}$
	\end{algorithmic}
\end{algorithm}

\subsection{Discussions: extensions to non-convex sets and other distance metrics}

When the pre-image set $\Cy$ is non-convex, the proposed cost is no longer convex. Furthermore, the projection in \eqref{eq:projection} is no longer necessarily unique, and as a consequence (following Theorem \ref{th:danskin}), the proposed cost in no longer differentiable. The algorithms proposed in this section could be extended to non-convex sets, however convergence to a global optimum would not be guaranteed, even for $\ell_1$-based sparse coding algorithms. For this reason, we focus in this paper on measurement functions with convex pre-image sets (which is the case for many measurement functions).

The proposed optimization problem \eqref{eq:SC_formulation} can be seen as a penalty method to solve the constrained problem \eqref{eq:constrained}. Here, the proposed cost used with an $\ell_2$-distance in \eqref{eq:distance} penalizes the samples outside of their consistency set with a quadratic penalty, while the samples inside the set have a cost of zero. However other types of distances have been used for penalty methods. An $\ell_1$-norm enforces a softer penalty on larger values, and can be used when large outliers are present. This has been proposed in the context of quantized measurements in \cite{2013_Jacques_Robust} and \cite{2016_Valsesia_Universal}. However, an $\ell_1$-norm leads to a non-differentiable data-fidelity term. Subgradient methods can be derived, however they are known to have a slow convergence rate \cite{2011_Bach_Convex}, and require careful tuning of the gradient descent parameter at every iteration. Other penalty functions include the Huber loss, which is differentiable and robust to outliers, but requires tuning of an additional parameter beforehand. Log-barrier functions can also be used to solve constrained problems in an unconstrained way, by forcing the estimates to remain interior to the set and away from the boundary \cite{1999_Bertsekas_Nonlinear}. Log-barriers function however tend to favour solutions that are far away from the boundary, which is not desirable in the context of consistent signal recovery. Finally, other approaches in the literature propose a cost based a maximum-likelihood estimation of the signal, assuming Gaussian additive noise \cite{2010_Zymnis_Compressed,2013_Bahmani_Robust,2014_Davenport_1}. However, they often involve computing cumulative distribution functions, and therefore lead to more complex formulations than the simple Euclidean distance proposed here. We focus here on the $\ell_2$-based data fidelity cost, since it leads to a simpler formulation and algorithms, and naturally extends the linear least-squares used in sparse coding and dictionary learning.

\section{Applications and link with previous work}
\label{sec:Applications}

In this section we show how the proposed framework can be applied to linear inverse problems such as denoising and inpainting, and nonlinear inverse problems such as declipping, de-quantization and 1-bit recovery. We give explicit formulations for the proposed cost and show links with costs proposed in the literature.

\subsection{Linear measurements:}

In the linear case $\y = \M\x$, the projection operator can be written as:
\begin{equation}\label{eq:linear_proj}
\Proj(\x) = \argmin_{\z} \|\x-\z\|_2^2 \quad \text{s.t.} \quad \M\z = \y,
\end{equation}
which (when $\M\M^T$ is invertible) can be computed as:
\begin{equation}
\Proj(\x) = \x - \M^T(\M\M^T)^{-1}(\M\x - \y).
\end{equation}
This shows that our proposed cost can be computed as:
\begin{equation}\label{eq:cost_linear}
\Lcal(\x,\y) = \frac{1}{2}\|\M^T(\M\M^T)^{-1}(\M\x-\y)\|_2^2,
\end{equation}
where we recognise the right pseudo-inverse $\M^T(\M\M^T)^{-1}$ of $\M$. The cost \eqref{eq:cost_linear} can thus be seen as ``inverting" the linear measurements and computing the error in the input space. In practice however computing the pseudo-inverse might be expensive or not feasible. In simple cases however, we retrieve classical costs used in the literature. When the measurements are clean or subject to additive Gaussian noise, $\M = \I$ and we have:
\begin{equation}
\Lcal(\x,\y) = \frac{1}{2}\|\x-\y\|_2^2,
\end{equation}
which is the classical linear least squares commonly used in sparse coding and dictionary learning. When $\M$ is a diagonal binary matrix (in the inpainting case), the projection \eqref{eq:linear_proj} can be computed as:
\begin{equation}
\Proj(\x) = \y + (\I - \M)\x,
\end{equation}  
and:
\begin{equation}
\Lcal(\x,\y) = \frac{1}{2}\|\y - \M\x\|_2^2,
\end{equation}
which is the masked least squares commonly used for signal inpainting \cite{2012_Adler_Audio}. This shows in particular that the algorithms proposed in Section \ref{sec:algorithms} extend classical algorithms such as ISTA, IHT or gradient-descent based dictionary learning algorithms.

\subsection{Saturated/clipped measurements}

In the case of saturated signals, using the notations of Section \ref{sec:background}, the feasibility set can be defined in closed form as:
\begin{align}
\begin{split}
\Cy = \{\x| \Mr& \y=\Mr \x, \Mcp \x \succeq \Mcp \y, \\
& \Mcm \x \preceq \Mcm \y\}
\end{split}
\end{align}
which is a convex set. The projection can be computed as:
\begin{align}
\begin{split}
\Pi_{\Cy}(\D\bfalpha) = \Mr \y + & \Mcp \max(\y,\D\bfalpha) \\
& + \Mcm \min(\y,\D\bfalpha).
\end{split}
\end{align}
This shows that the proposed cost can be written in closed form as:
\begin{equation}\label{eq:clipping_consistency}
\begin{split}
& \Lcal(\D\bfalpha,\y) = \frac{1}{2} \big[\|\Mr (\y-\D\bfalpha)\|_2^2 \\& + \|\Mcp (\y-\D\bfalpha)_+\|_2^2 +  \|\Mcm (\y-\D\bfalpha)_-\|_2^2\big].
\end{split}
\end{equation}
The proposed cost thus generalizes the soft consistency metric proposed in \cite{2011_Laska_Democracy,2013_Kitic_Consistent,2014_Siedenburg_Audio} for declipping. When $\Psi(\bfalpha) = \|\bfalpha\|_0$, Algorithm \ref{alg:SC} is thus equivalent to the consistent IHT declipping algorithm proposed in \cite{2013_Kitic_Consistent}. When $\Psi(\bfalpha) = \|\bfalpha\|_1$, Algorithm \ref{alg:SC} is similar to the ISTA-type declipping algorithms proposed in \cite{2014_Siedenburg_Audio}.

\subsection{Quantized measurements}
We consider a general quantizer defined by quantization levels $y_i$ and quantization sets $f^{-1}\{y_i\} = [l_i, u_i)$ for each sample $i$. As commented earlier and discussed in \cite{2009_Dai_Distortion,2010_Zymnis_Compressed}, one can assume $f^{-1}\{y_i\} = [l_i, u_i]$ (the closure of $[l_i, u_i)$) without affecting the cost function. The projection operator for each sample $x_i$ can be computed as:
\begin{align}
\Proj(x_i) = 
\begin{cases}
u_i & \text{if} \quad x_i \geq u_i\\
l_i & \text{if} \quad x_i \leq l_i\\
x_i & \text{otherwise.}
\end{cases}
\end{align}
When concatenating the quantization boundaries $\l = [l_1,...,l_N]$ and $\u = [u_1,...,u_N]$, the proposed cost can be written as:
\begin{equation}\label{eq:cost_quantization}
\Lcal(\D\bfalpha,\y) = \frac{1}{2} \big[\|(\l-\D\bfalpha)_-\|_2^2+\|(\u-\D\bfalpha)_+\|_2^2\big].
\end{equation}
For example in the case of uniform mid-riser quantizer:
\begin{equation}
\Lcal(\D\bfalpha,\y) = \frac{1}{2}\Big[\|(\y-\frac{\Delta}{2}-\D\bfalpha)_-\|_2^2+\|(\y+\frac{\Delta}{2}-\D\bfalpha)_+\|_2^2\Big].
\end{equation}
This cost is somewhat similar to the optimization problem in \cite{2011_Dai_Information} solved at every iteration. However in \cite{2011_Dai_Information}, both the signal and its projection are optimized, solving a quadratic programming problem at every iteration. Here, only the sparse coefficients need to be optimized using simple gradient descent. Interestingly, the authors in \cite{2011_Dai_Information} also assume convexity of the quantization sets, in order to ensure the solutions to be well-defined and unique. The idea of projecting onto convex sets for image decoding was also used in \cite{1998_Thao_Set}.

\subsection{1-bit sensing}

In the case of 1-bit measurements, the projection operator can easily be computed for each sample $x_i$ as:

\begin{align}
\Proj(x_i) = 
\begin{cases}
x_i & \text{if} \quad \sign(x_i) = y_i\\
0 & \text{otherwise,}
\end{cases}
\end{align}
and it can be easily verified that:
\begin{equation}
\Lcal(\D\bfalpha,\y) = \frac{1}{2}\|\big(\y\odot(\D\bfalpha)\big)_-\|_2^2,
\end{equation}
which shows that the proposed cost is equivalent to the cost \eqref{eq:1bit_l2} proposed for 1-bit signals \cite{2008_Boufounos_1,2009_Boufounos_Greedy}. We can verify that this indeed corresponds to the clipping consistency cost, since when $\theta = 0^+$ we have in \eqref{eq:clipping_consistency}:
\begin{equation}
\begin{split}
\Lcal(\D\bfalpha,\y) & = \frac{1}{2}\big[\|\Mcp (\0^+-\D\bfalpha)_+\|_2^2\\
& \qquad +  \|\Mcm (\0^--\D\bfalpha)_-\|_2^2\big]\\
& = \frac{1}{2}\|\big(\sign(\y)\odot(\D\bfalpha)\big)_-\|_2^2
\end{split}
\end{equation}
Similarly for quantization, taking $l_i = 0$, $u_i \rightarrow +\infty$ if $y_i>0$, $l_i \rightarrow -\infty, u_i = 0$ if $y_i<0$ in \eqref{eq:cost_quantization}, gives the same result.

\subsection{Summary}

The proposed framework unifies cost functions used for denoising, inpainting, declipping and 1-bit recovery. It can also be used for quantization with any quantization map. In particular, performing sparse coding or dictionary learning on clipped, quantized or 1-bit measurements can be done as simply and computationally efficiently as with clean, linear measurement. Experimental results are shown in the next section.

\section{Evaluation}
\label{sec:evaluation}

We evaluate the performance of the proposed framework on declipping, de-quantization and 1-bit recovery tasks. Each algorithm is evaluated in terms of Signal-to-Noise Ratio (SNR) signal: $\text{SNR}(\hat{\x},\x) = 20 \log\frac{\|\x\|_2}{\|\x-\hat{\x}\|_2}$ where $\hat{\x} = \hat{\D}\hat{\bfalpha}$ is the estimated signal, and $\x$ is the reference clean signal. However in the case of 1-bit measurements, since the signal can only be recovered up to an amplitude factor, we use the \emph{angular} SNR \cite{2013_Jacques_Quantized}:
\begin{equation}
\text{SNR}_{\text{angular}}(\hat{\x},\x) \triangleq 20 \log\frac{\|\x\|_2}{\|\x-\frac{\|\x\|_2}{\|\hat{\x}\|_2}\hat{\x}\|_2}.
\end{equation}

\subsection{Consistent sparse coding}

We first evaluate the performance of the proposed consistent sparse coding algorithms. We generate a dictionary $\D \in \mathbb{R}^{32 \times 64}$, with i.i.d. normally distributed entries, and unit $\ell_2$-norm columns. We then generate $T = 2000$ $K$-sparse coefficients $\bfalpha_t \in \mathbb{R}^{64}$ with i.i.d normal distribution for the coefficients. We normalize the resulting signals $\x_t = \D \bfalpha_t$ to unit $\ell_{\infty}$ norm, and artificially clip or quantize the signals as $\y_t = f(\x_t)$. We consider clipping with different levels $\theta$. For quantization, we consider a uniform mid-rise quantizer that quantizes the input space $[-1,1]$ using $N_b$ bits, i.e. using $2^{N_b}$ quantization levels of size $\Delta = 2/2^{N_b}$.

Figure \ref{fig:proposed_vs_ADMM} shows the performance of the proposed algorithms, implemented with $\Psi(\bfalpha) = \|\bfalpha\|_1$, $\lambda = 10^{-2}$ (for Algorithm \ref{alg:SC}), and the adaptive parameter strategy (Algorithm \ref{alg:SC_paramfree}). For the adaptive strategy, we used $\lambda^0 = \|\nabla\Lcal(\D\bfalpha,\y)\big\rvert_{\bfalpha = \0}\|_{\infty} = \|\D^T\Pi_{\Cy}(\0)\|_{\infty}$, and $\lambda^{k+1} = \lambda^k/2$. The algorithm is stopped when a consistency level $\Lcal(\D\bfalpha,\y) \leq \epsilon$ with $\epsilon = 10^{-3}$ is reached. We also compare with the ADMM-based algorithm that solves the exact constraint problem \eqref{eq:constrained}, as representative of the state-of-the-art. This was recently used for de-quantization in \cite{2016_Moshtaghpour_Consistent}, and declipping in \cite{2015_Kitic_Sparsity}, although here we use the $\ell_1$-norm unlike the $\ell_0$-norm used in \cite{2015_Kitic_Sparsity}\footnote{An \emph{analysis} sparsity formulation was also proposed in \cite{2015_Kitic_Sparsity}, but we refer here to the synthesis version.}. The ADMM algorithm is limited to 400 iterations due its computational complexity. Other algorithms are run for a maximum of 400 iterations for a fair comparison. Figure \ref{fig:proposed_vs_ADMM} shows that the proposed algorithm with a carefully chosen fixed parameter $\lambda$ can outperform the exact constrained-based algorithm. This suggests that solutions that are slightly less consistent but sparser, might lead to better reconstruction. However fixing the parameter $\lambda$ might not be optimal, since in practice the best parameter might depend on the sparsity level, distortion level and/or signal energy. We can see in Figure \ref{fig:proposed_vs_ADMM} that using Algorithm \ref{alg:SC_paramfree} and the above described sequence of $\lambda^k$ performs better than using a fixed parameter $\lambda$. In particular the solution reached by the proposed adaptive algorithm performs better than the constrained approach. Another important consideration is the computational time. The ADMM-based algorithm involves computing non-orthogonal projections at each iteration, which have to be computed iteratively, resulting in a high overall computational cost. The proposed algorithms \ref{alg:SC} and \ref{alg:SC_paramfree} on the other hand only require gradient computations and element-wise operations, and are thus computationally efficient. The average computational time of each algorithm is shown in Table \ref{table:comp_time}.

\makeatletter%
\if@twocolumn%
\begin{figure}[h!]
	\centering
	\begin{subfigure}[h!]{0.48\columnwidth}
		\includegraphics[width=1\textwidth]{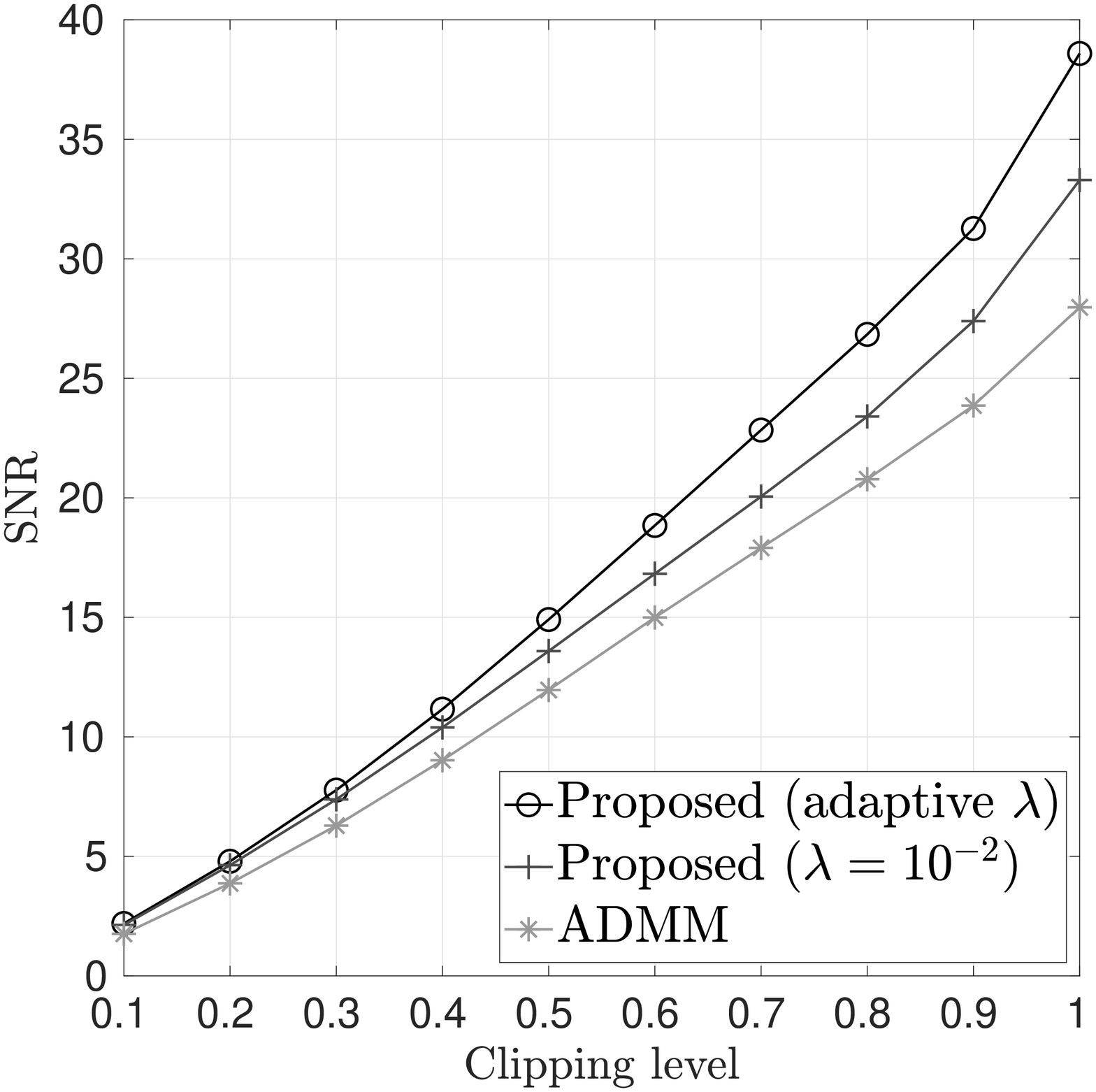}
		\caption{Clipping}
	\end{subfigure}
	~
	\begin{subfigure}[h!]{0.48\columnwidth}
		\includegraphics[width=\textwidth]{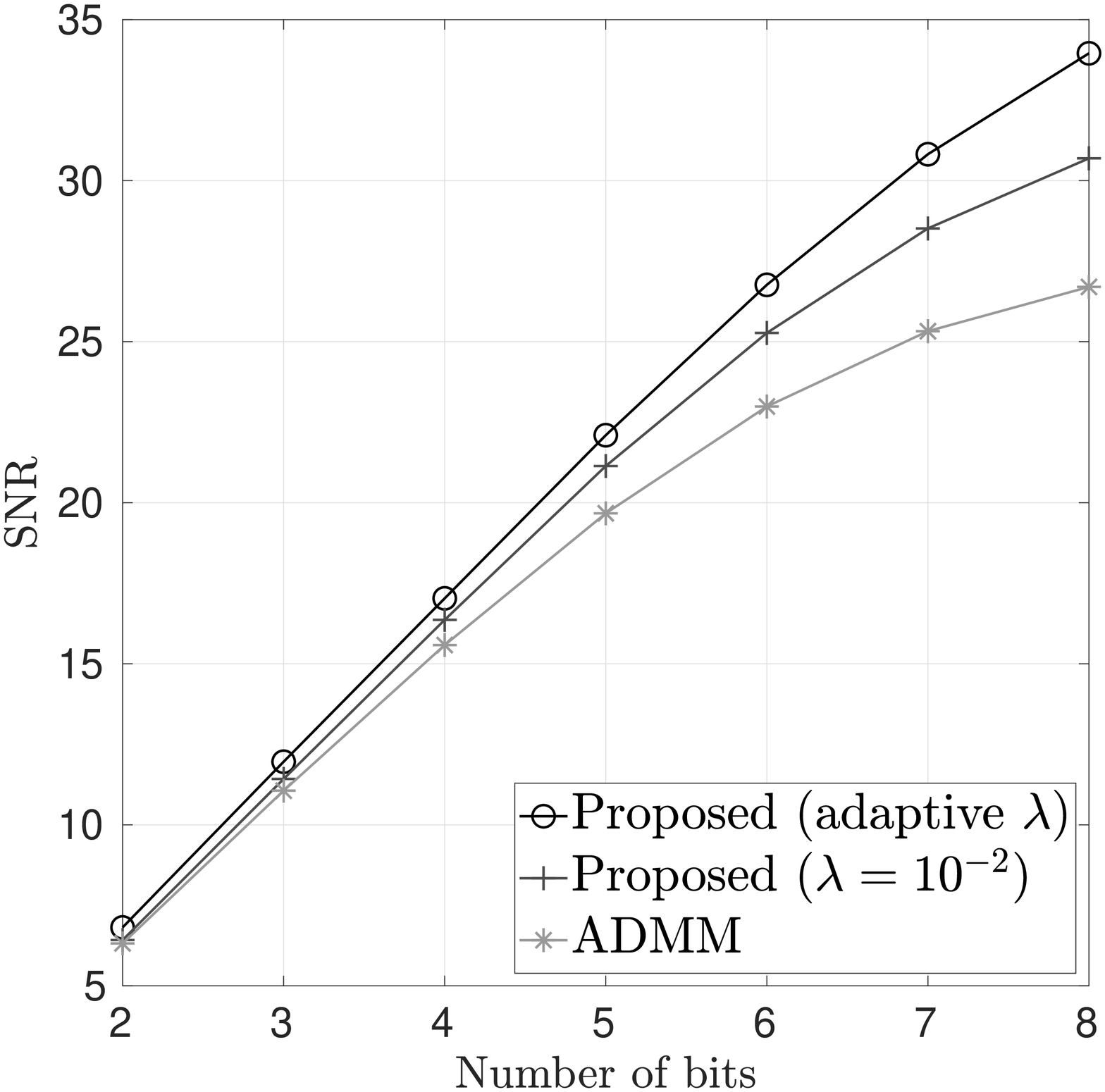}
		\caption{Quantization}
	\end{subfigure}
	\caption{Comparison of the proposed consistent sparse coding algorithms using the $\ell_1$-norm, versus solving the constrained problem with ADMM.}
	\label{fig:proposed_vs_ADMM}
\end{figure}
\else
\begin{figure}[h!]
	\centering
	\begin{subfigure}[h!]{0.4\columnwidth}
		\includegraphics[width=1\textwidth]{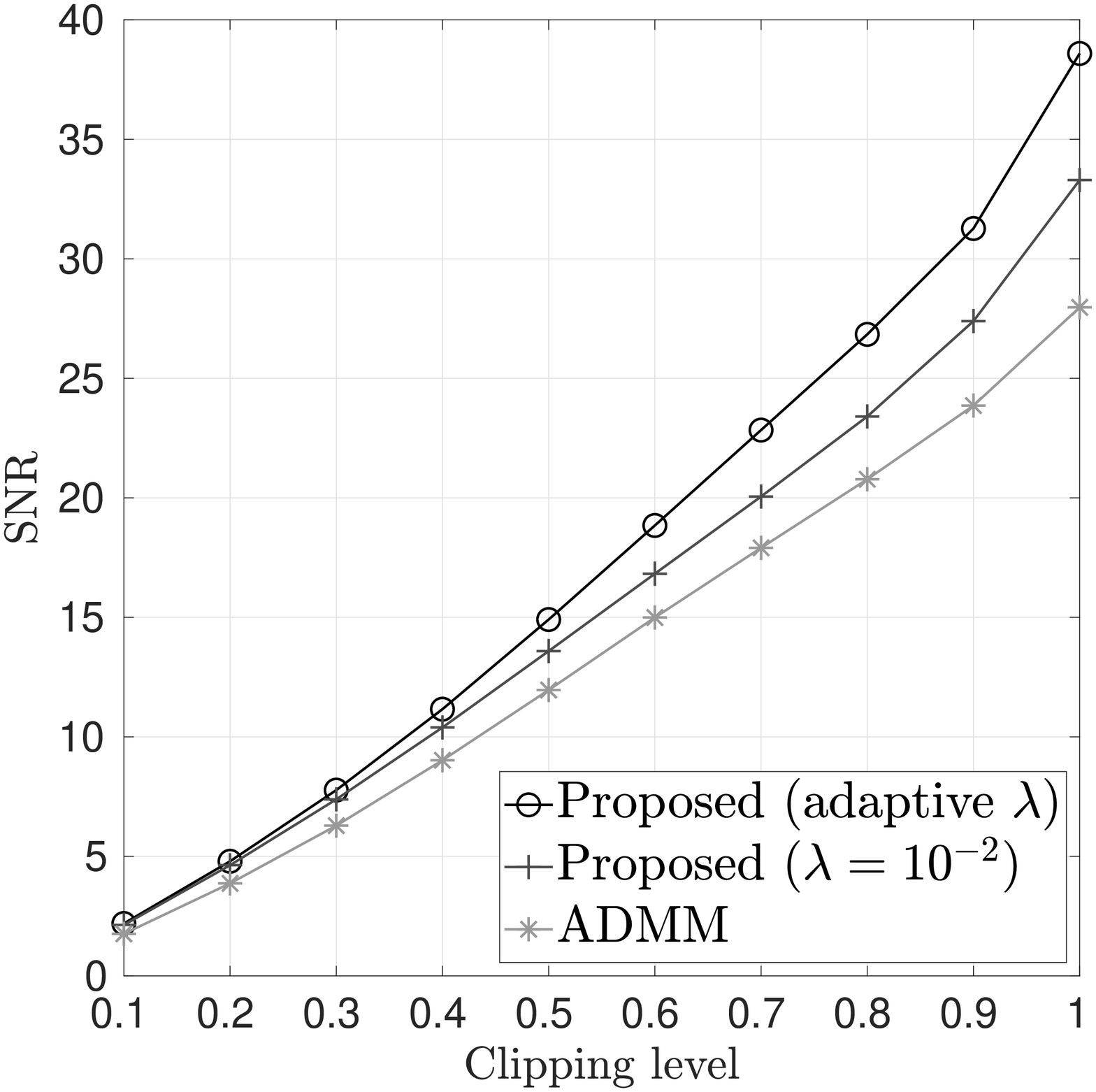}
		\caption{Clipping}
	\end{subfigure}
	~
	\begin{subfigure}[h!]{0.4\columnwidth}
		\includegraphics[width=\textwidth]{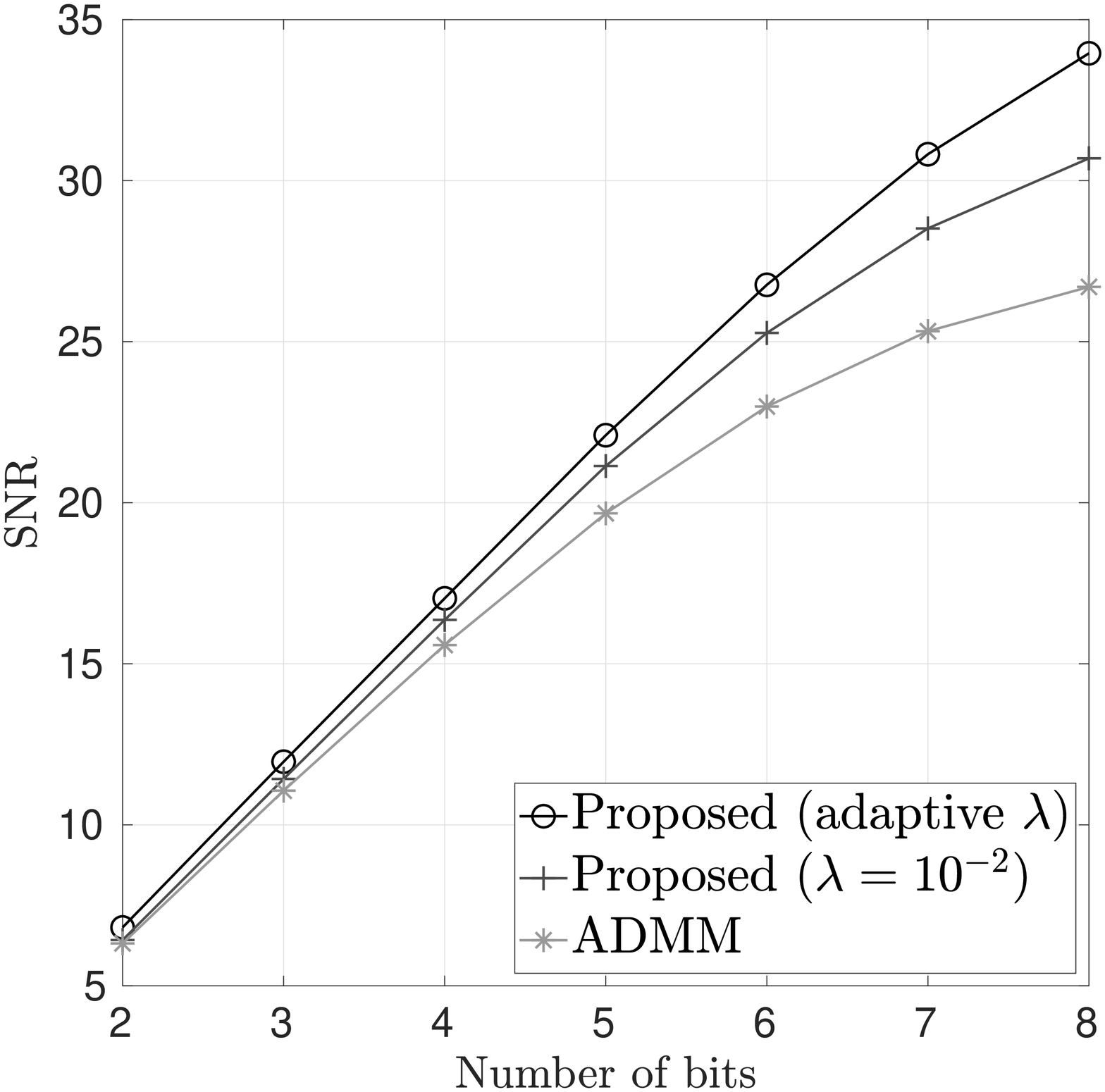}
		\caption{Quantization}
	\end{subfigure}
	\caption{Comparison of the proposed consistent sparse coding algorithms using the $\ell_1$-norm, versus solving the constrained problem with ADMM.}
	\label{fig:proposed_vs_ADMM}
\end{figure}
\fi
\makeatother

\begin{table}[h!]
	\centering
	\begin{tabular}{|c||c|c|c|}
		\hline
		cpu time (s) & Alg. \ref{alg:SC_paramfree} & Alg. \ref{alg:SC} & ADMM\\
		\hline\hline
		declipping & \textbf{1.08} & 2.75 & 281.7\\
		\hline
		dequantization & \textbf{1.19} & 2.91 & 320.3\\
		\hline
	\end{tabular}
	\caption{Average computational time of each algorithm}
	\label{table:comp_time}
\end{table}

\subsection{Consistent dictionary learning}


\begin{table*}[!t]
\renewcommand{\arraystretch}{1.3}
\caption{Performance of consistent dictionary learning on 1-bit recovery}
\label{table:consDL}
\centering
\begin{tabular}{ |c||c|c|c| c|}
 \hline
 Angular SNR (dB) & Classical & Classical &Consistent sparse coding (DCT)& Consistent dictionary learning\\
& sparse coding (DCT) &  dictionary learning & (Algorithm \ref{alg:SC}) & (Algorithm \ref{alg:DL}) \\
 \hline\hline
 Female speech & 5.69 & 5.50 & 5.91 & \textbf{6.12}\\
  \hline
 Male speech & 4.67 & 4.47 & 4.50 & \textbf{4.70}\\
  \hline
\end{tabular}
\end{table*}

We evaluate the proposed consistent dictionary learning algorithm on real speech signals. The dataset consists of 10 male and female speech signals, taken from the SISEC dataset \cite{2017_Liutkus_2016}. Each signal is 10s long, sampled at 16kHz, with 16 bits per sample. Each signal is normalized to unit $\ell_{\infty}$ norm, and then processed using overlapping time frames of size $N = 256$, with rectangular windows and 75\% overlap, for a total of approximately $T = 2500$ frames per signal. All sparse coding experiments are run using an overcomplete  DCT dictionary of size $M = 512$. All dictionary learning algorithms are initialized using the same DCT dictionary. To speed up convergence, the sparse coefficients and dictionaries are initialized at every iteration using the estimates from the previous iteration. Similarly, the step size parameters $\mu_1$ and $\mu_2$ can be re-estimated at every iteration as $\mu_1 = 1/\|\D\|_2^2$ and $\mu_2 = 1/\|\A\|_2^2$ using the current estimates of $\D$ and $\A$.

Figure \ref{fig:exp_audio} shows the performance of the proposed framework for consistent sparse coding and dictionary learning, compared to classical linear sparse coding and dictionary learning. In the case of declipping, the classical approach is to discard the clipped samples and treat declipping as an inpainting problem \cite{2012_Adler_Audio}. In the quantization case, the classical approach is to treat the quantized signals as noisy signals with variance $\frac{\Delta^2}{12}$ \cite{2009_Boufounos_Greedy}, which is enforced by stopping the algorithm when an error $\epsilon = \frac{\Delta^2}{12}$ is reached. For 1-bit signals, we simply use the sign measurements directly as the input \cite{2009_Boufounos_Greedy}. All algorithms are run with an $\ell_0$-constraint with fixed $K = 32$. In all three cases classical sparse coding is run using IHT \cite{2009_Blumensath_Iterative}, and classical dictionary learning alternates between IHT to update the sparse coefficients and gradient descent to update the dictionary. We perform 50 iterations for the sparse coding algorithms, and 50 iterations (with 20 iterations at each inner step) for dictionary learning. 

Figure \ref{fig:audio_declipping} shows the declipping performance, for different clipping levels ranging from $\theta = 0.1$ (highly clipped) to $\theta = 1$ (unclipped). Figure \ref{fig:audio_declipping} demonstrates several things: First, using measurement consistency greatly improves the reconstruction. Consistent sparse coding shows an improvement of up to 8dB compared to classical sparse coding. Consistent dictionary learning shows up to 10dB improvement compared to classical dictionary learning. This improvement is greater when the signals are highly distorted ($\theta \leq 0.5$). As expected, the two frameworks give equivalent results when $\theta = 1$, which shows how the proposed consistent framework naturally extends classical sparse coding. Second, we can see that consistent dictionary learning greatly improves the reconstruction performance compared to using a fixed DCT dictionary, since it adapts the dictionary to the signals of interest. This shows that the learned dictionary, although it has been learned using only the low-energy unclipped samples, along with consistency penalties, generalizes well to the unobserved clipped samples. In particular, the proposed consistent dictionary learning algorithm outperforms all the other methods. Finally, it is interesting to point out that when the signals are highly clipped ($\theta = 0.1)$, classical dictionary learning does not improve compared to classical sparse coding with DCT. This is probably due to a lack of data to learn from, since most of the data is clipped and discarded. Our consistent dictionary learning algorithm on the other hand, makes use of the clipped data, and is able to learn and improve the performance by 1.7dB.

Figure \ref{fig:audio_dequant} shows the results for quantization, from highly quantized ($N_b = 2$ bits) to lightly quantized ($N_b =8$). Similarly, we can see that using measurement consistency improves the performance when the signals are heavily quantized $N_b \leq 5$. As expected, the performance of the consistent framework and of the classical one are comparable when the signals are lightly distorted ($N_b = 8$). Dictionary learning improves the performance (compared to using a fixed dictionary), and the proposed consistent dictionary learning algorithm outperforms the other methods.

\makeatletter%
\if@twocolumn%
\begin{figure}[h!]
	\centering
	\begin{subfigure}[h!]{\columnwidth}
		\includegraphics[width=\textwidth]{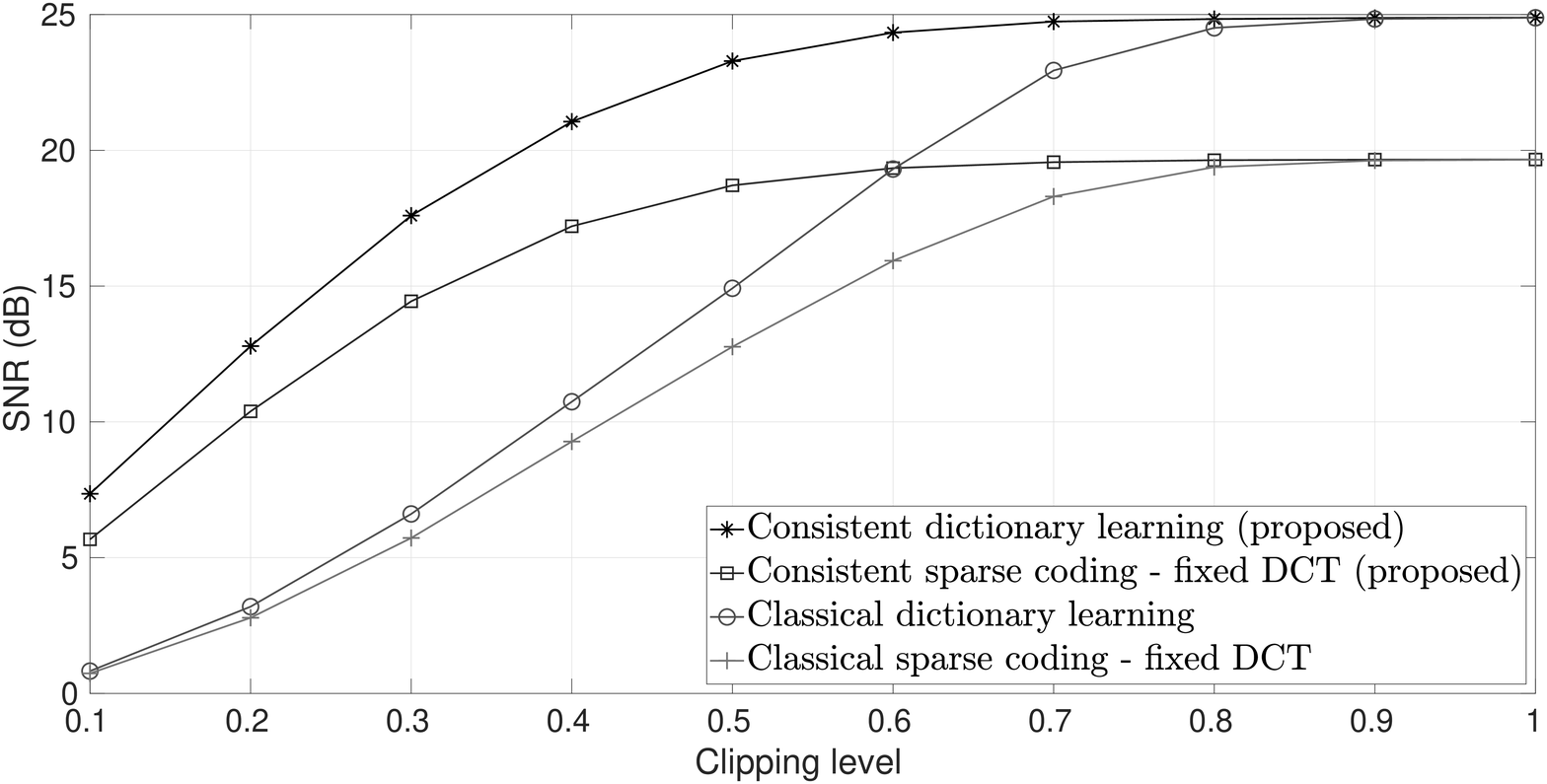}
		\caption{Declipping}
		\label{fig:audio_declipping}
	\end{subfigure}
	~
	\begin{subfigure}[h!]{\columnwidth}
		\includegraphics[width=\textwidth]{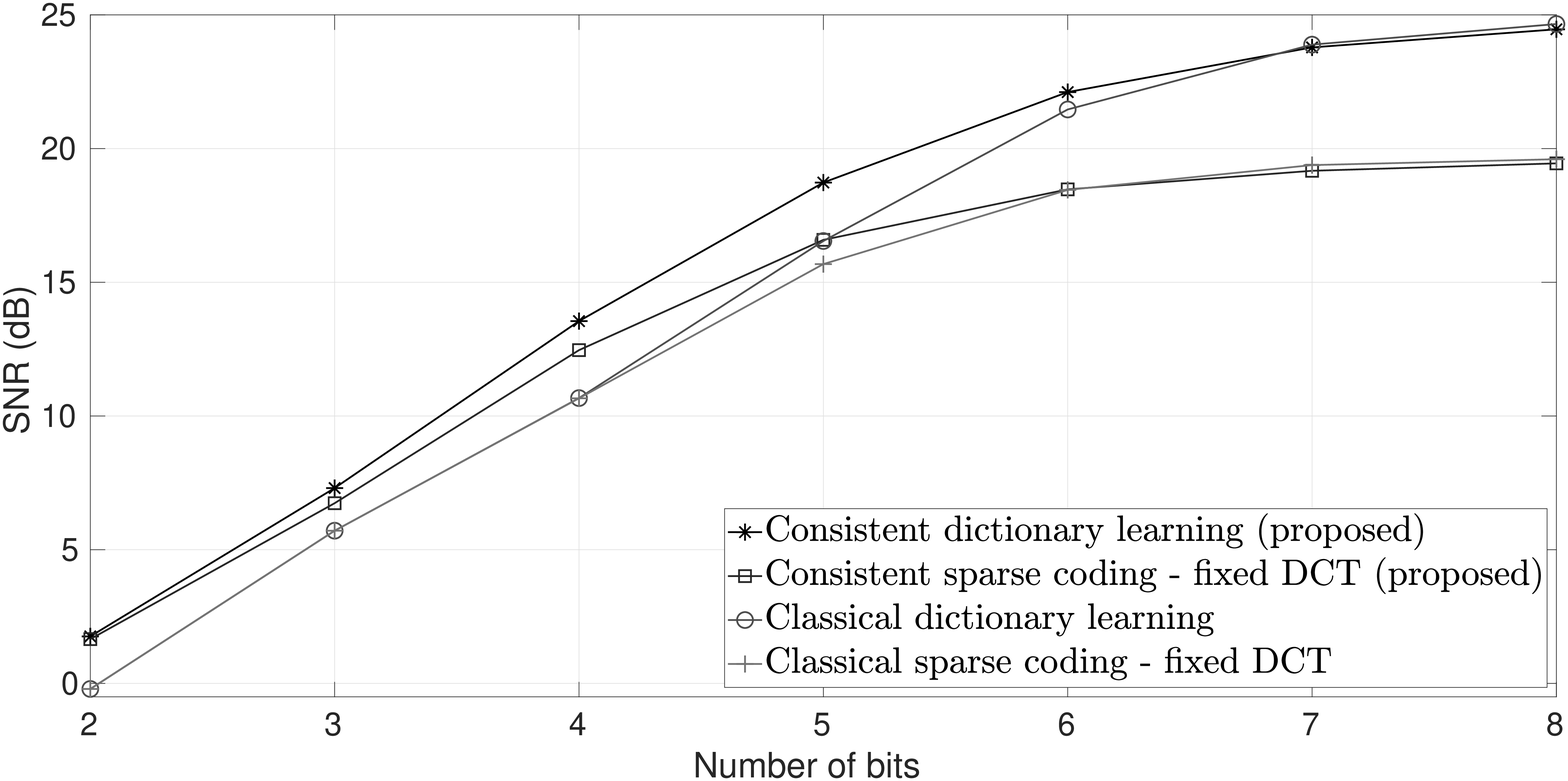}
		\caption{De-quantization}
		\label{fig:audio_dequant}
	\end{subfigure}
	\caption{Comparison of the proposed consistent sparse coding and dictionary learning algorithms, compared to classical sparse coding and dictionary learning.}
	\label{fig:exp_audio}
\end{figure}
\else
\begin{figure}[h!]
	\centering
	\begin{subfigure}[h!]{0.7\columnwidth}
		\includegraphics[width=\textwidth]{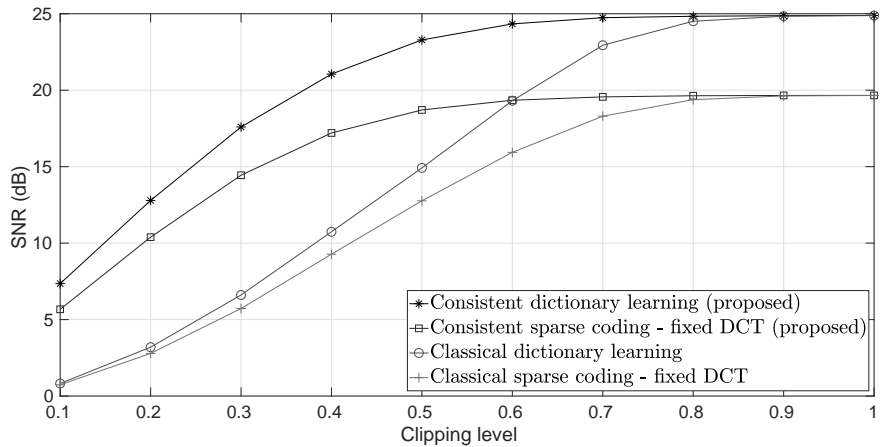}
		\caption{Declipping}
		\label{fig:audio_declipping}
	\end{subfigure}
	~
	\begin{subfigure}[h!]{0.7\columnwidth}
		\includegraphics[width=\textwidth]{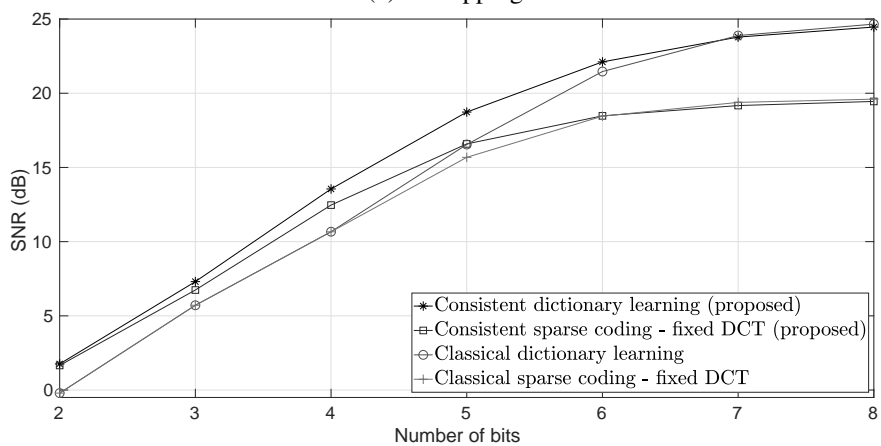}
		\caption{De-quantization}
		\label{fig:audio_dequant}
	\end{subfigure}
	\caption{Comparison of the proposed consistent sparse coding and dictionary learning algorithms, compared to classical sparse coding and dictionary learning.}
	\label{fig:exp_audio}
\end{figure}
\fi
\makeatother

The results for 1-bit data are shown in Table \ref{table:consDL}. The results for sparse coding and consistent sparse coding are comparable (note that here for simplicity we don't enforce the signal to be on the unit circle, unlike in \cite{2008_Boufounos_1,2009_Boufounos_Greedy}). Classical dictionary learning here performs worse than classical sparse coding with a fixed dictionary, presumably because the reconstructed signal overfits the $\pm 1$ sign measurements. This shows that classical dictionary learning does not perform well with 1-bit data. The proposed consistent dictionary learning however, outperforms consistent sparse coding and classical dictionary learning.


\section{Conclusion}
\label{sec:conclusion}

We have presented a unified framework for signal reconstruction from certain types of nonlinear measurements such as clipping, quantization and 1-bit measurements. We proposed a cost function that takes into account prior knowledge about the measurement process, by minimizing the distance to the pre-image of the received signal. When the pre-image is convex, we have shown that the proposed cost is a convex, smooth and Lipschitz differentiable, which makes it ideal for proximal-based algorithms. The proposed cost generalizes the linear least-squares commonly used in sparse coding and dictionary learning, as well as cost functions proposed for declipping and 1-bit recovery. We proposed proximal based sparse coding and dictionary learning algorithms, which naturally extend classical algorithms and can deal with clipped, quantized and 1-bit measurements.

\appendices

\section{Projection Theorem}
\label{app:projection_theorem}

We recall the following theorem \cite[Prop. B.11]{1999_Bertsekas_Nonlinear}:
\begin{theorem}[Projection Theorem \protect{\cite[Prop. B.11]{1999_Bertsekas_Nonlinear}}]\label{th:projection_theorem}
	Let $\mathcal{C}$ be a closed convex set in $\mathbb{R}^N$. Then, the following hold:
	
	a) For every $\x \in \mathbb{R}^N$, there exists a unique $\z^* \in \mathcal{C}$ such that $\z^*$ minimizes $\|\x-\z\|_2$ over all $\z \in \mathcal{C}$. $\z^*$ is called the \emph{projection} of $\x$ onto $\mathcal{C}$ and is noted $\Pi_{\mathcal{C}}(\x)$. In other words:
	\begin{equation}\label{eq:projection}
	\Pi_{\mathcal{C}}(\x) \triangleq \argmin_{\z \in \mathcal{C}} \|\x-\z\|_2.
	\end{equation}
	
	b) For $\x \in \mathbb{R}^N$, $\z^* = \Pi_{\mathcal{C}}(\x)$ if and only if:
	\begin{equation}
	(\z-\z^*)^T(\x-\z^*) \leq 0 \quad \forall \z \in \mathcal{C}.
	\end{equation}
	
	c) $\x \mapsto \Pi_{\mathcal{C}}(\x)$ is continuous and non-expansive, i.e:
	\begin{equation}
	\|\Pi_{\mathcal{C}}(\x_1)-\Pi_{\mathcal{C}}(\x_2)\|_2 \leq \|\x_1- \x_2\|_2 \quad \forall \x_1, \x_2 \in \mathbb{R}^N.
	\end{equation}
\end{theorem}

\section{Danskin's Min-Max Theorem}
\label{app:Danskin}

\begin{theorem}[Danskin's Min-Max Theorem \protect{\cite[Section 4.1]{1998_Bonnans_Optimization}}, \cite{1967_Danskin_theory}]\label{th:danskin}
	Let $\mathcal{C}$ be a compact\footnote{Note that compactness is only required to ensure existence of a minimum, according to Weierstrass' theorem.} set, and $g(\x) = \min_{\z \in \mathcal{C}} \phi(\x,\z)$. Suppose that for each $\z \in \mathbb{R}^N$, $\phi(\cdot,\z)$ is differentiable with gradient $\nabla_{\x} \phi(\x,\z)$, and $\phi(\x,\z)$ and $\nabla_{\x} \phi(\x,\z)$ are continuous on $\mathbb{R}^N \times \mathbb{R}^N$. Define $\mathcal{Z}(\x) = \argmin_{\z \in \mathcal{C}} \phi(\x,\z)$. Then $g$ is directionally differentiable, with derivative in the direction $\h$:
	\begin{equation}
	\nabla g(\x;\h) = \min_{\z \in \mathcal{Z}(\x)} \nabla_{\x} \phi(\x,\z)^T\h \quad \forall \h,
	\end{equation}
	In particular, when the minimum is attained at a unique point ($\mathcal{Z}(\x) = \{\z^*\}$), $g$ is differentiable with gradient:
	\begin{equation}
	\nabla g(\x) = \nabla_{\x} \phi(\x, \z^*).
	\end{equation}
\end{theorem}
In other words, Danskin's Min-Max theorem says that if the minimum over a family of continuous and continuously differentiable functions is attained at a unique point $\z^*$, then the gradient of the minimum over this family of functions can be computed by simply evaluating that gradient at the optimum $\z^*$.

\section{Proofs of Proposition 1-3 and 5}
\label{app:proofs}

\begin{proof}[Proof of Proposition \ref{prop:convexity}:]
$\Lcal(\cdot,\y)$ is a minimum of a family of convex functions $(\x,\z) \mapsto \frac{1}{2}\|\x-\z\|_2^2$ over a non-empty convex set, so by \cite[Section 3.2.5]{2004_Boyd_Convex}, $\Lcal(\cdot,\y)$ is convex.
\end{proof}

\begin{proof}[Proof of Proposition \ref{prop:smoothness}:]
$\Lcal(\x, \y) = \min_{\z \in \Cy} \dEU(\x, \z)$ with $\dEU(\x,\z) = \frac{1}{2}\|\x-\z\|_2^2$. For all $\z \in \mathbb{R}^N, \dEU(\cdot,\z)$ is differentiable with gradient $\nabla_{\x} \dEU(\x,\z) = \x-\z$. Furthermore, $\dEU(\x,\z)$ and $\nabla_{\x}\dEU(\x,\z)$ are continuous in $(\x,\z)$, and $\mathcal{Z}(\x) = \argmin_{\z \in \Cy} \dEU(\x,\z)$ is uniquely defined as $\mathcal{Z}(\x) = \{\Pi_{\Cy}(\x)\}$ by the Projection Theorem. Using Danskin's Min-Max theorem, we can then conclude that:
\begin{equation}
\begin{split}
\nabla_\x \Lcal(\x,\y) = & \nabla_{\x} \dEU(\x, \Pi_{\Cy}(\x))\\
= & \x-\Pi_{\Cy}(\x).
\end{split}
\end{equation}
\end{proof}

\begin{proof}[Proof of Proposition \ref{prop:lipschitz}:]
Let $\x_1, \x_2 \in \R^N$. By the projection theorem, we have:
\begin{equation}
\begin{split}
(\Proj(\x_2)-\Proj(\x_1))^T(\x_1-\Proj(\x_1)) & \leq 0\\
(\Proj(\x_1)-\Proj(\x_2))^T(\x_2-\Proj(\x_2)) & \leq 0.
\end{split}
\end{equation}
Adding and rearranging these two equations gives:
\begin{equation}
\label{eq:ineq}
\begin{split}
\|\Proj(&\x_1)- \Proj(\x_2)\|_2^2  \\ \leq & (\Proj(\x_1)-\Proj(\x_2))^T(\x_1-\x_2).
\end{split}
\end{equation}
We can then show that:
\begin{equation}
\begin{split}
\|&\nabla_{\x}\Lcal(\x_1, \y)-\nabla_{\x}\Lcal(\x_2, \y)\|^2_2 \\
& = \|\x_1-\Proj(\x_1)-(\x_2-\Proj(\x_2))\|^2_2\\
& = \|\x_1-\x_2\|^2_2+\|\Proj(\x_1)-\Proj(\x_2)\|^2_2\\
& \qquad -2(\x_1-\x_2)^T(\Proj(\x_1)-\Proj(\x_2))\\
& \leq \|\x_1-\x_2\|^2_2-\|\Proj(\x_1)-\Proj(\x_2)\|^2_2\\
& \leq  \|\x_1-\x_2\|_2^2
\end{split}
\end{equation} 
where we have used \eqref{eq:ineq} in the third line.
\end{proof}

\begin{proof}[Proof of Proposition \ref{prop:paramfree}]
	The proof is given assuming exact estimation of the $\bfalpha^k$ coefficients, however similar results for approximate estimates are also available \cite{1999_Bertsekas_Nonlinear}.
	We have by definition of $\bfalpha^k$ and $\bfalpha^{k+1}$:
	\begin{equation}\label{eq:proof}
	\begin{cases}
	&F(\lambda^{k+1}, \bfalpha^{k+1}) \leq F(\lambda^{k+1}, \bfalpha^{k})\\
	&F(\lambda^{k}, \bfalpha^{k}) \leq F(\lambda^{k}, \bfalpha^{k+1})
	\end{cases}
	\end{equation}
	Summing and rearranging these two equations gives:
	\begin{equation}
	(\lambda^{k+1} - \lambda^{k})(\Psi(\bfalpha^{k+1})-\Psi(\bfalpha^{k})) \leq 0,
	\end{equation}
	Since $\lambda^{k+1} < \lambda^k$, we have $\Psi(\bfalpha^{k+1}) \geq \Psi(\bfalpha^{k})$, from which follows (by combining with e.g., the first line of \eqref{eq:proof}) $\Lcal(\D\bfalpha^{k+1},\y) \leq \Lcal(\D\bfalpha^k,\y)$.
	
	We have furthermore that:
	\begin{equation}
	\min_{\bfalpha} F(\lambda^k,\bfalpha) \leq \min_{\D\bfalpha \in \Cy} F(\lambda^k,\bfalpha) = \lambda^k \Psi^*
	\end{equation}
	Where $\Psi^*$ is the optimum value of \eqref{eq:constrained}. In other words:
	\begin{equation}\label{eq:limit}
	0 \leq \Lcal(\D\bfalpha^{k},\y) + \lambda^k\Psi(\bfalpha^{k}) \leq \lambda^k \Psi^*
	\end{equation}
	When $\{\lambda^k\} \rightarrow 0$, taking the limit in \eqref{eq:limit} shows that $\lim_k \Lcal(\D\bfalpha^{k},\y) = 0$ and $\lim_k \Psi(\bfalpha^{k}) \leq \Psi^*$. In particular, $F(\lambda^k, \bfalpha^k)$ converges to the optimum value $\Psi^*$ of the constrained problem \eqref{eq:constrained}.
\end{proof}

\ifCLASSOPTIONcaptionsoff
  \newpage
\fi



%

\bibliographystyle{IEEEtran}
\bibliography{IEEEabrv,nonlinearDL.bib}

\end{document}

%% file: clipping_func.tikz
\begin{tikzpicture}[scale=0.5]

\draw[->] (-4,0) -- (4,0) node[above] {$x_i$};
\draw	(0,0) node[below right] {0}
		(2,0) node[below] {$\theta^+$}
		(-2,0) node[below left] {$\theta^-$}
		(0,2) node[left] {$\theta^+$}
		(0,-2) node[below left] {$\theta^-$};

\draw[->] (0,-4) -- (0,4) node[above] {$y_i$};
\draw[dotted] (2,0) -- (2,2);
\draw[dotted] (0,2) -- (2,2);
\draw[dotted] (-2,0) -- (-2,-2);
\draw[dotted] (0,-2) -- (-2,-2);

\draw[thick] (-4,-2) -- (-2,-2) -- (2,2) -- (4,2);

\end{tikzpicture}

%% file: quantization_func.tikz
\begin{tikzpicture}[scale=0.5]

\newcommand{\var}{1.25}

\draw[->] (-4,0) -- (4,0) node[above] {$x_i$};;
\draw[->] (0,-4) -- (0,4) node[above] {$y_i$};
\draw	(0,0) node[below right] {0}
		(1*\var,0) node[anchor=north] {$\Delta$}
		(2*\var,0) node[anchor=north] {$2\Delta$}
		(-1*\var,0) node[below left] {$-\Delta$}
		(-2*\var,0) node[below left] {$-2\Delta$}
		(0,1/2*\var) node[anchor=east] {$\frac{\Delta}{2}$}
		(0,3/2*\var) node[anchor=east] {$\frac{3\Delta}{2}$}
		(0,5/2*\var) node[anchor=east] {$\frac{5\Delta}{2}$}
		(0,-1/2*\var) node[below left] {$-\frac{\Delta}{2}$}
		(0,-3/2*\var) node[below left] {$-\frac{3\Delta}{2}$};

\draw[thick] (0,1/2*\var) -- (1*\var,1/2*\var)
			 (1*\var,3/2*\var) -- (2*\var,3/2*\var)
			 (2*\var,5/2*\var) -- (3*\var,5/2*\var)
			 (-1*\var,-1/2*\var) -- (0,-1/2*\var)
			 (-1*\var,-3/2*\var) -- (-2*\var,-3/2*\var)
			 (-2*\var,-5/2*\var) -- (-3*\var,-5/2*\var);
			 
\draw[dotted] (\var,0) -- (1*\var,3/2*\var)
			 (2*\var,0) -- (2*\var,5/2*\var)
			 (-\var,-3/2*\var) -- (-1*\var,0)
			 (-2*\var,-5/2*\var) -- (-2*\var,0); 
			 
\draw[dotted] (0,3/2*\var) -- (1*\var,3/2*\var)
			  (0,5/2*\var) -- (2*\var,5/2*\var)
			  (-1*\var,-3/2*\var) -- (0,-3/2*\var)
			  (-2*\var,-5/2*\var) -- (0,-5/2*\var); 
			  
\fill[black] (0,1/2*\var) circle (2pt);
\fill[black] (1*\var,3/2*\var) circle (2pt);
\fill[black] (2*\var,5/2*\var) circle (2pt);
\fill[black] (-1*\var,-1/2*\var) circle (2pt);
\fill[black] (-2*\var,-3/2*\var) circle (2pt);
\fill[black] (-3*\var,-5/2*\var) circle (2pt);

\end{tikzpicture}

%% file: 1bit_func.tikz
\begin{tikzpicture}[scale=0.5]

\draw[->] (-4,0) -- (4,0) node[above] {$x_i$};;
\draw[->] (0,-4) -- (0,4) node[above] {$y_i$};
\draw	(0,0) node[below right] {0}
		(0,3) node[anchor=east] {$+1$}
		(0,-3) node[above left] {$-1$};

\draw[thick] (0,3) -- (4,3)
			 (-4,-3) -- (0,-3);
			  
\fill[black] (0,3) circle (3pt);

\end{tikzpicture}